\documentclass[journal, twoside, web]{ieeecolorpreprint}
\usepackage{generic}
\usepackage{cite}
\usepackage{amsmath,amssymb,amsfonts}
\usepackage{graphicx}
\usepackage{algorithm,algorithmic}
\usepackage{hyperref}
\usepackage{textcomp}
\def\BibTeX{{\rm B\kern-.05em{\sc i\kern-.025em b}\kern-.08em\kern-.1667em\lower.7ex\hbox{E}\kern-.125emX}}
\usepackage{multicol}
\usepackage{float}
\usepackage{xcolor}
\usepackage{optidef}
\usepackage{makecell}
\usepackage{tikz}
\usetikzlibrary{shapes,arrows,positioning,quotes,fadings}
\tikzstyle{block} = [draw, rectangle, minimum height=2em, minimum width=3em]
\tikzstyle{virtual} = [coordinate]
\usepackage{tikzpeople}
\usepackage{svg}
\usepackage{subcaption}
\usepackage{enumerate}

\newtheorem{prop}{Proposition}
\newtheorem{thm}{Theorem}
\newtheorem{cor}{Corollary}
\newtheorem{lem}{Lemma}
\newtheorem{prob}{Problem}
\newcounter{subprob}[prob]
\renewcommand{\thesubprob}{\theprob\alph{subprob}}
\newenvironment{subprob}{
  \refstepcounter{subprob}
  \par\noindent\textit{Problem \thesubprob:}\space\ignorespaces
}{}
\newtheorem{defi}{Definition}
\newtheorem{rem}{Remark}
\newtheorem{asm}{Assumption}

\DeclareMathOperator{\Ker}{Ker}
\DeclareMathOperator{\enc}{Enc}
\DeclareMathOperator{\dec}{Dec}
\DeclareMathOperator{\vecc}{Vec}

\title{Anomaly Detection with LWE Encrypted Control} 
\author{Rijad Alisic, Junsoo Kim, and Henrik Sandberg
\thanks{This research was funded in part by the Swedish Foundation for Strategic Research through the CLAS project (RIT17-0046), by the Seoul National University of Science and Technology, by the Swedish Research Council (2016-00861, 2023-04770), and by the Swedish Civil Contingencies Agency (CERCES2).}
\thanks{R.~Alisic and H.~Sandberg are with the Division of Decision and Control Systems at KTH Royal Institute of Technology, Sweden {\tt \small (e-mail: rijada@kth.se, hsan@kth.se)}.}
\thanks{J.~Kim is with the Department of Electrical and Information Engineering, Seoul National University of Science and Technology, Korea {\tt \small (e-mail: junsookim@seoultech.ac.kr)}.}
}

\begin{document}

\maketitle



\begin{abstract}
    Detecting attacks using encrypted signals is challenging since encryption hides its information content. We present a novel mechanism for anomaly detection over Learning with Errors (LWE) encrypted signals without using decryption, secure channels, nor complex communication schemes. Instead, the detector exploits the homomorphic property of LWE encryption to perform hypothesis tests on transformations of the encrypted samples. The specific transformations are determined by solutions to a hard lattice-based minimization problem. While the test's sensitivity deteriorates with suboptimal solutions, similar to the exponential deterioration of the (related) test that breaks the cryptosystem, we show that the deterioration is polynomial for our test. This rate gap can be exploited to pick parameters that lead to somewhat weaker encryption but large gains in detection capability. Finally, we conclude the paper by presenting a numerical example that simulates anomaly detection, demonstrating the effectiveness of our method in identifying attacks.
\end{abstract}


\section{Introduction} \label{sec:introduction}
Recent interest in the security of control systems has been focused on understanding the learning-based attacker, who uses data to perform stealthy or undetectable attacks~\cite{muller2018,park2016}. The data is obtained from disclosure attacks on the cyber component that connects the plant to a computational element used to control it, thus creating a cyber-physical system (CPS). Additionally, due to the advancements in data-driven methods for control~\cite{Coulson2019,proctor2016}, recent works show that attackers do not need explicit models to generate undetectable attacks. Rather, they can compute such attacks directly from a plant's input-output data, a relatively unexplored topic, see~\cite{wang2019,alisic2021ecc,taheri2021,adachi2023}. Many undetectable attacks further require the system to be in a particular state~\cite{TEIXEIRA2015,weerakkody2015,gheitasi2022}, forcing the attacker to correctly estimate the state. 


Limiting the information an attacker can extract from the data limits its attack-learning capabilities. One of two common methods is to target the attacker's estimator by injecting noise into the data transmitted over the cyber component~\cite{sandberg2015,cortes2016,Hassan2020}. However, introducing noise into the CPS opens up for \emph{stealthy} attacks; detectable attacks that remaing hidden by exploiting the detector's reduced sensitivity, a necessary setting to manage otherwise high false-alarm rates~\cite{giraldo2018}. Encryption is another commonly used tool to limit the information an attacker could extract from a signal~\cite{kim2022dynamic}. Traditional encryption methods, such as Advanced Encryption Standardization (AES)~\cite{AES2001}, can be cumbersome to implement in a control-specific system because a hand-shake mechanism must be established, where the potential exchange of credentials, such as decryption, evaluation, or authentication keys, must be carried out safely~\cite{kogiso2018}, thus introducing another attack vector. 
These solutions rely on the controller actively participating in the system's security without getting corrupted. Furthermore, not all services are beneficent, such as the honest-but-curios actor~\cite{Boenisch2023}. Thus, traditional methods only ensure data confidentiality during transmission, which leaves computational resources, like controllers, vulnerable to attackers that could collude with or directly hack into them to extract information.

\begin{figure}
\centering
    {\begin{tikzpicture}[pin distance=1cm,>=stealth,auto, node distance=1cm,rotate=90]

\node [virtual] (model) {};

\node [block] (system) {Plant $G$};

\node[virtual, above = 41pt of model](ysplit){};
\node[block, right = 15pt of ysplit] (enc) {Enc};

\node (devi) [right=54pt of model,scale=0.1] {\includesvg{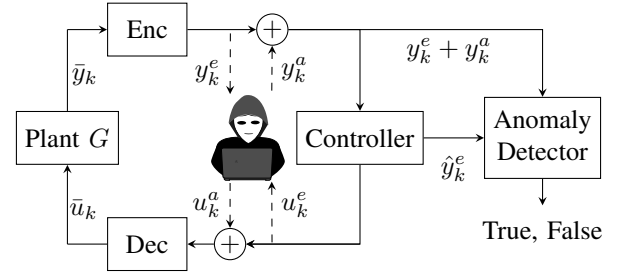}};

\node[block,right = 4pt of devi](operator){Controller};

\node[virtual, below = 30pt of system](usplit2){};
\node[block, right = 15pt of usplit2] (dec) {Dec};

\node[virtual, right = 23pt of dec](attu){};
\node[virtual, right = 23pt of enc](atty){};
\node[block, right = 75pt of devi](anom){\makecell{Anomaly\\ Detector}};

\node[circle, draw, inner sep=0.02cm ,right = 10pt of dec](adddow){$+$};
\node[virtual, right = 43pt of adddow](usplit){};

\node[circle, draw, inner sep=0.02cm ,right = 26pt of enc](addup){$+$};

\node[virtual, right = 27pt of addup](upp){};

\node[block, below = 10pt of anom, draw opacity=0](anomout){${\text{True, False}}$};

\draw [->] (anom.south) -- (anomout.north){};

\draw[->](ysplit) -- (enc){};
\draw[->](enc) -- (addup){};

\draw[->] (addup.east) |- node[yshift=-8pt, xshift=-55pt] {$y^e_k+y^a_k$} (anom.north){};
\draw[->] (operator.east) -- node[yshift=-20pt] {$\hat{y}^e_k$} (anom.west){};

\draw [-] (system) -- node[yshift=-2pt, xshift=15pt] {$\bar y_k$} (ysplit){};
\draw [->] (addup.east) |-  (operator.north){};

\draw [->] (operator.south) -|  (adddow.east){};

\draw[->] (usplit) -- (adddow.east){};
\draw[->] (adddow) -- (dec.east){};
\draw[-] (dec) |- (usplit2){};
\draw [->] (usplit2) --  node[yshift=-2pt, xshift=15pt] {$\bar u_k$} (system){};
\draw [dashed, <-] (addup) -- node {$y_k^a$} ([yshift = -2ex] devi.north){};
\draw [dashed, <-] ([yshift = -2ex]devi.south) -- node[yshift=3pt] {$u_k^e$} ([yshift = -2ex] attu){};

\draw [dashed, <-] ([yshift = 1.5ex]devi.north) -- node[yshift=-4pt] {$y_k^e$} ([yshift = 1.5ex] atty){};
\draw [dashed, <-] (adddow) -- node {$u_k^a$} ([yshift = 1.5ex] devi.south){};
\end{tikzpicture}}
     \caption{An attacker can inject stealthy attack signals $y^a_k$ and $u^a_k$ into the loop, by observing $\bar y_k$ and $\bar u_k$. While encryption hinders the attacker from these observations, it also obstructs the detection of anomalies. We propose an anomaly detector that uses {$y^e_k+y_k^a$} and $\hat{y}^e_k$ for detection through encryption.} \label{fig:adv_attack__vector6}
\end{figure}

Homomorphic encryption (HE) allows computations to be performed directly on encrypted data, with the results then carrying over post-decryption. Thanks to this property, HE can \emph{also protect the data during computation}. Several forms of HE exist, including the Paillier cryptosystem capable of performing addition and multiplication on encrypted messages~\cite{paillier1999}. In this paper, we consider a cryptosystem based on the Learning with Errors (LWE) problem~\cite{regev2009}, which is claimed to be \emph{post-quantum secure}. Unlike traditional cryptosystems that are based on \emph{worst-case hard} problems, which are efficiently solved by quantum computers, the LWE problem is a lattice problem that is believed to be \emph{average-case hard}~\cite{ajtai1996}.

Combining HE, or other encryption solutions, with anomaly detection has generally been intractable due to the need for elaborate encryption, key sharing, and authentication schemes. Without authentication, for instance, an attacker could ignore detectability since the encryption hides the attack from detection unless the detector can establish elaborate communication schemes with the plant~\cite{alexandru2022}, such as multiple rounds of secure communication with the plant~\cite{buns2018}, or securely sharing keys~\cite{alisic2023modelfreelwe}. We want to highlight that the recent paper~\cite{alexandru2022} seems to be the first work to treat this problem for dynamical CPS, which is why the literature on this problem is sparse. Our paper follows the same line of work. However, we present a novel approach to anomaly detection over encrypted signals that do not need elaborate communication schemes. In particular, our contributions address the following problem:

\begin{informalProblem}
    Can the LWE cryptosystem, which enables homomorphic encrypted control, simultaneously enable anomaly detection without revealing the messages? Specifically, we want to know:
    \begin{enumerate}[a)]
        \item What computations are required (from the controller) to realize and enable anomaly detection?
        \item Are there any trade-offs regarding detection power, detection time, or encryption strength that must be considered?
    \end{enumerate}
\end{informalProblem}

Note that our problem differs from the typical approach in the literature of modifying the communication scheme. Instead, we ask how to exploit the properties of LWE to obtain a signal that can be used for anomaly detection while keeping the standard communication scheme in control systems. To answer these questions, we rely on the property that homomorphism allows us to transform encrypted messages into other, valid, encrypted messages. First, we seek a transformation of the plant output that produces an encrypted \emph{residual} that we can use for hypothesis testing. Creating a residual requires (some) knowledge about the plant dynamics and controller. We can embed such knowledge into the controller design, typically done through a Kalman filter. We will use the controller's filter to predict the next output and use it to form a sequence of encrypted residual messages. Second, we seek transformations of these encrypted residual messages so that the resulting message becomes independent of the encryption keys, allowing us to perform anomaly detection. The anomaly detection takes the form of a hypothesis test that does not reveal information about the secret vector of the cryptosystem, thus preserving the original LWE security.

\subsubsection*{Contribution}
The main contributions of this paper is a transformation of samples that leads to a tractable detection scheme, which:
\begin{enumerate}
    \item \label{point:consistent} produces an output that allows for a hypothesis test, while \emph{not} revealing the secret vector (Theorem~\ref{thm:main_theorem_rejection_rule}),
    \item \label{point:better_properties} the test has better statistical power than the related test that reveals the secret vector (Theorem~\ref{thm:bounds}),
\end{enumerate}
We want to emphasize~\ref{point:consistent}), in the sense that we are not breaking the cryptosystem nor solving the LWE problem~\cite{regev2009}. Rather, by reviewing~\cite{regev2009}, we discovered that their algorithm uncovered the secret vector by performing a hypothesis test that rejects the uniform distribution in favor of an alternative one. We base anomaly detection on the opposite test; rejecting the alternative distribution in favor of the uniform one. This subtle difference leads to a test with better statistical properties, \ref{point:better_properties}).

Unsurprisingly, the statistical power of the detection test relies on solving a lattice problem, namely finding the shortest vector (smallest $\Vert \cdot \Vert_2$) in a lattice. An approximate solution (within an exponentially growing bound) is given by the  Lenstra-Lenstra-Lov\'asz (LLL) algorithm~\cite{lenstra1982}, which takes basis vectors as inputs and tries to output the shortest vectors by finding a set of (nearly) orthogonal vectors in the lattice. Although the LLL algorithm runs in polynomial time, its output depends on the order of the initial vectors, which makes the search space combinatorial already before the algorithm is applied. In this paper, however, the detection mechanism does not require the shortest vector. Instead, a longer, but \textit{sufficiently short}, vector suffices to perform detection, and the statistical power of the anomaly detection increases if shorter vectors are found. However, since the LLL algorithm runs in polynomial time, we argue that efficient anomaly detection can be devised over encrypted signals.



The outline of the paper is as follows. We set up some preliminary notation and introduce the LWE-based encryption scheme in Section~\ref{sec:prelim}. We define our problem statement and detection mechanism in Section~\ref{sec:problem}. Our main results are presented in Section~\ref{sec:results}, where the anomaly detection essentially boils down to a binary hypothesis test. In Section~\ref{sec:controllers}, we show how to realize the detection mechanism. Trade-off considerations are discussed in Section\ref{sec:vectors}. Finally, we discuss the subtle difference between our hypothesis test, and the one set up by Regev to prove the hardness of the LWE problem in Section~\ref{sec:hyptest}. We conclude the paper with numerical examples in Section~\ref{sec:numerical} and conclusions in Section~\ref{sec:conclusions}.



\section{Preliminaries}\label{sec:prelim}
In this section, we shall introduce the tools we need to define LWE-based encryption and the related problems we need to consider to enable anomaly detection. Encryption is typically defined over a discrete set of numbers, so let us define the number space we will work with. Let the rounding operator be denoted by $\lfloor \cdot \rceil$, and the floor operator by $\lfloor \cdot \rfloor$, where $\lfloor a \rfloor$ is the largest integer smaller than $a \in \mathbb R$.

\begin{defi}
    Denote the set of integers modulo $q$ by: $$  \mathbb{Z}_q:= \left\{i\in\mathbb{Z}  
    \left \vert -\frac{q}{2}   \leq i <\frac{q}{2} \right. \right \}.$$
\end{defi}

Integers outside the set $\mathbb{Z}_q$ must be mapped back onto it. We achieve this mapping by the modulo operation.
\begin{defi}
    We define the modulus operation as $$  a\bmod q := a - \left\lfloor\frac{a+ \frac{q}{2} }{q}\right\rfloor q, $$ so that each integer $a\in \mathbb{Z}$ maps onto $\mathbb{Z}_q$. 
\end{defi}
The space of $n$-dimensional vectors, where each element is in $\mathbb{Z}_q$ is denoted by $ \mathbb{Z}_q^n$, and similarly for the ${n \times m}$-dimensional matrices by $\mathbb{Z}_q^{n\times m}$.

We require two distributions defined over $\mathbb{Z}_q$ to formulate the LWE encryption scheme. The first one is the discrete uniform distribution, where $x \sim  \mathcal{U}_{q}$ implies that $\mathrm{Pr}(X=x) = \frac{1}{q}, \quad \text{for } x \in \mathbb Z_q$. The second one is the discrete normal distribution, defined as follows.

\begin{defi}\label{def:discretenormal}
    The discrete normal distribution, $\mathcal{N}_\mathcal{S} (\mu, \sigma^2)$, is defined over a support set $\mathcal{S}$, with the probability distribution:
    \begin{equation*}
        \mathrm{Pr}_\mathcal{S}(X=x)=\frac{\mathrm{e}^{-\frac{(x-\mu)^2}{2\sigma^2}}}{\sum \limits_{y \in \mathcal{S}}\mathrm{e}^{-\frac{(y-\mu)^2}{2\sigma^2}}}, \quad \text{for } x \in \mathcal{S}.
    \end{equation*}
\end{defi}
We will only work with $\mu = 0$, the distribution's mean for odd $q$. Similarly, $\sigma^2$ is the variance of the distribution. The encrypted messages, essentially linear combinations of noise and plaintext messages, can be considered samples from a distribution with the plaintext message $\mu$.



\subsection{LWE-based Encryption}

The encryption scheme we consider is based on the LWE framework~\cite{regev2009}, which, in addition to the encrypted message $\tilde m$, discloses a random public matrix $P$ that has been used in the encryption. Loosely, the encryption operator outputs two quantities when applied to the message $m$, $\enc^\lambda_s(m)=(P, \tilde m)\coloneqq m^e$, where $\lambda$ is a set of parameters that defines the \emph{security level}; see Definition~\ref{def:lwe}, and the \emph{secret vector} $s$ is used for encryption and decryption. The tuple $(P, \tilde m)$ is communicated to the computation service, while the secret vector $s$ is kept at the plant. The public matrix $P$ is needed, in addition to the message $\tilde m$ to uncover the plaintext message: $\dec^\lambda_s(P,\tilde m)=m$.

The term \emph{homomorphic} refers to the property that operations, such as addition or multiplication, on the encrypted signals carry over to the decrypted signals. To stabilize LTI systems, a feedback controller typically needs the ability to add signals together and multiply them with scalars. Therefore, an encryption scheme that allows for stabilizing controllers needs the following property:
\begin{equation}\label{eq:decrypt_homomoprh}
\begin{aligned}
\dec^\lambda_s\left ( \left(b_1 \odot \enc^\lambda_s\left (m_1\right ) \right) \oplus \left (b_2 \odot \enc^\lambda_s \left(m_2\right )\right)  \right) \\ = b_1m_1+b_2m_2,
\end{aligned}
\end{equation}
for $b_1, b_2 \in \mathbb{Z}$ and where $\oplus$ and $\odot$ refer to operations on the encrypted signals that carry over as addition and multiplication, respectively, in the plaintext signals. 

Encryption based on the LWE problem is defined as follows.
\begin{defi}\label{def:lwe}
The message $m \in \mathbb Z_q^p$ is encrypted with $\lambda=(v,r,\sigma^2,q)$-level of security using the scheme
\begin{equation}\label{eq:lweencmsg}
    \enc^\lambda_s (m) = (P, Ps+mr+e \bmod q) \coloneqq m^e,
\end{equation}
where the elements of $P \in \mathbb{Z}_q^{p \times v}$ and $s\in \mathbb{Z}_q^v$ are sampled from $\mathcal{U}_q$. The $i$th element of the error vector, $e \in \mathbb{Z}_q^p$, is sampled from the discrete normal distribution, $e(i) \sim \mathcal{N}_{\mathbb{Z}_q}(0,\sigma^2)$, $\forall i$ and is independent from the other elements. The scaling number $r>0$ separates the plaintext message $m$ and the error $e$, see Remark~\ref{rem:roundingerror}.
\end{defi}


Note that the secret vector $s$ is generated once and is used for all messages, which is necessary for the homomorphic property of the LWE-based scheme. The decryption operator also uses the secret vector to decrypt the message.
\begin{defi}
The input signals to the physical system are decrypted using the scheme
\begin{equation}\label{eq:decrypt}
    \dec^\lambda_s(m^e) =\left \lfloor \frac{Ps+mr+e \bmod q-Ps \bmod q}{r} \right \rceil.
\end{equation}
\end{defi}
We will drop the scripts $\lambda$ and $s$ in the rest of the paper, as they will be fixed. Note that the decryption of individual messages might fail because of two reasons, shown next.
\begin{lem}\label{lem:succ_decrypt}
    The decryption operator in~\eqref{eq:decrypt} returns the correct message if and only if $\vert mr + e \vert < \frac{q}{2}$ and $\vert e \vert <\frac{r}{2}$.
\end{lem}

\begin{rem}~\label{rem:roundingerror}
The error $e$ must be small relative to $m$ to not create a faulty decryption. As evident by Lemma~\ref{lem:succ_decrypt}, a synthetic separation is achieved by scaling the message with $r>0$ in~\eqref{eq:lweencmsg} so that $\frac{\vert e\vert }{r}$ becomes less than $\frac{1}{2}$. 
\end{rem}

The (partially) homomorphic property~\eqref{eq:decrypt_homomoprh} for the LWE-based encryption is realized by ordinary addition of the tuples.
\begin{lem}\label{lem:ishomomorphic}
    The LWE-based cryptosystem can be homomorphic with respect to addition. If $\vert m_1r+m_2r+e_1+e_2 \vert < \frac{q}{2},$ and {$\vert {e_1+e_2} \vert <\frac{r}{2},$} then, $\dec(m_1^e+m_2^e \bmod q)= m_1+m_2,$ where $\enc(m_1)=m_1^e$ and $\enc(m_2)= m_2^e$.
\end{lem}

Lemmas~\ref{lem:succ_decrypt} and ~\ref{lem:ishomomorphic} are easy to prove by applying the decryption operator to (sums of) encrypted messages and keeping track of potential overflows due to the $\bmod q$ operation. Note that multiplications with (plaintext) integer scalars can be realized by summing the message repeatedly
\begin{equation*}
    c m = \dec(cm^e \bmod q)=\dec \left( \sum \limits_{i=1}^c m^e \bmod q \right),
\end{equation*}
where $c \in \mathbb Z_q$, which will give the correct decryption if $\vert cmr+ce \vert < \frac{q}{2} $ and $\vert ce \vert < \frac{r}{2}$.

The conditions for decryption will always have a non-zero probability of failure. Specifications about the underlying system need, therefore, to be given early on, such as the largest message $m_k$ that might appear, so that appropriate parameters are chosen. These parameters, $q$ and $r$ to be specific, help to determine the security level $\lambda$ in Definition~\ref{def:lwe}, and subsequently, the probability of failure upon decryption. They were furthermore shown to affect the probability that a learning-based attacker would succeed with their attack in a previous work~\cite{alisic2023modelfreelwe}.

\subsection{Lattices}
LWE-based cryptosystems operate over integer messages and integer-linear combinations of them. Lattices represent such systems well.
\begin{defi}
A lattice $\mathbb L$ is the set of points given by
\begin{equation*}
    \mathbb L \coloneqq \left \{ a_1v_1 + a_2v_2 + \ldots + a_nv_n \, \vert \, a_i\in \mathbb Z, v_i \in \mathcal L \subset \mathbb Z^{m} \right \}.
\end{equation*}
The lattice's basis, $\mathcal L$, can be represented as a matrix $L \in \mathbb Z_q^{m \times n}$, whose columns are the basis vectors $v_i$.
\end{defi}
The anomaly detector we define will operate over the following class of \textit{$q$-ary lattices}.
\begin{defi}\label{def:qaryLattice}
    Let $\Ker_q V$ for matrix $V \in \mathbb{Z}_q^{n\times m}$ denote $$ \Ker_q V:= \left \{v \in\mathbb{Z}_q^m \vert Vv = 0 \bmod q\right \}.$$
\end{defi}
This set can be computed (with polynomial complexity) using the Hermite Normal Form (HNF)~\cite{kannan1979}. 

When working with the discrete normal distribution over lattices, the variance parameter, $\sigma^2$, must be sufficiently large. Specifically, if $\sigma > \eta_\epsilon(\mathbb L)$, where $\eta_\epsilon(\mathbb L)$ is the \emph{smoothing parameter} for lattice $\mathbb L$, the discrete normal distribution approximates the continuous normal distribution within a statistical distance of $\mathcal O (\epsilon)$~\cite{boneh2011}. 
For a sum of samples from the discrete normal distribution to behave like a continuous normal distribution, the prior variance of the samples must accommodate this. Specifically, if $X \sim \mathcal{N}_{\mathbb{Z}_q} (\mu_x, \sigma^2_x)$ and $Y \sim \mathcal{N}_{\mathbb{Z}_q} (\mu_y, \sigma^2_y)$, then for $X+Y$ to be approximately sampled as $\mathcal{N}_{\mathbb{Z}_q} (\mu_x+\mu_y, \sigma^2_x+\sigma^2_y)$, we require that $\sigma_x^2, \sigma_y^2 \geq 2\eta_\epsilon^2(\mathbb L)$. Furthermore, a too large $\sigma^2$ may lead to a considerable weight of the resulting sum falling outside of the modulus space, $\mathrm{Pr}(\Vert c \Vert_2|X|> \frac{q}{2}) \not \approx 0$, which would induce an erroneous decryption. Therefore, will need make sure that $\sigma^2$ is not too large either, as shown in the following:
\begin{asm}\label{asm:noise_var}
    The noise variance is sufficiently large to accommodate a specified number of additions and multiplications of the signals, quantified by the vector $c $, so that $ \frac{q^2}{4} \frac{\Omega^2_\epsilon}{\Vert c \Vert_2^2} > \sigma^2 > \Vert c \Vert^2_2 \eta_\epsilon^2(\mathbb L)$, for some $0<\epsilon<1$. 
\end{asm}


The $\eta_\epsilon(\mathbb L)$ in Assumption~\ref{asm:noise_var} is sufficient to ensure the property that the resulting sum can be considered as having been independently sampled from a discrete normal distribution~\cite{boneh2011}. Furthermore, we pick a $\Omega_\epsilon<1$ so that decrypting the sum fails with a low probability, $\mathrm{Pr}(\Vert c \Vert_2|X|> \frac{q}{2}) \leq \mathcal{O}(\epsilon)$.

\section{Problem Formulation} \label{sec:problem}
We are now ready set up the problems we solve in this paper. Consider the attacked Plant $G$ shown in Fig.~\ref{fig:adv_attack__vector6}, which we model as
\begin{equation} \label{eq:system}
    G: \begin{cases}
    \begin{aligned}
    x_{k+1} & =  Ax_k + B u_k + w_k, \\
      y_k & =   Cx_k + Du_k +v_k.
    \end{aligned}
    \end{cases}
\end{equation}
This a discrete-time system where $x_k \in\mathbb{R}^n$ is the state, $u_k\in\mathbb{R}^m$ is the input, $y_k\in\mathbb{R}^p$ is the output, and ${A, \, B, \, C, \, D}$ are matrices of appropriate sizes with the elements being in $\mathbb{R}$. The process noise $w_k\in \mathbb{R}^n$ and the measurement noise $v_k\in \mathbb{R}^p$ are assumed to be (continuous) Gaussian, zero-mean, i.i.d random variables with $\mathbb{E} [ w_k w_k^\top]=\Sigma_w$ and $\mathbb{E} [ v_k v_k^\top]=\Sigma_v$. 

While LWE-based encryption only works with integers, physical plants typically output real values. Therefore, a conversion between reals and integers is required. We will use a quantization through a standard rescaling and rounding,
\begin{equation}\label{eq:signal_rescaling}
    \bar y_k = \lceil c y_k \rfloor, \quad u_k = \frac{1}{c} \bar u_k,
\end{equation}
where the desired level of accuracy is determined by a scalar $c >1$ (a treatment of such a \emph{uniform quantizer}, is given in~\cite{delchamps1990}). Due to the transformations~\eqref{eq:signal_rescaling}, the output signal can be encrypted, and the input signal can be converted back to a real value after decryption.

Quantization introduces precision limitations of the considered signals. Typically, any value below the precision level is ignored, introducing an error that could be made arbitrarily small by choosing large parameters in~\eqref{eq:signal_rescaling} and a sufficiently large message representation (such as a 64-bit binary sequence). These errors propagate similarly during post-computation. However, \emph{informally}, the LWE encryption shifts the plaintext message to some other uniformly random chosen value in the message space. The small errors incurred by ignoring values below the precision threshold propagate as significant errors in the deciphered message. Even a slight change in the ciphertext message leads to substantial differences post-decryption.

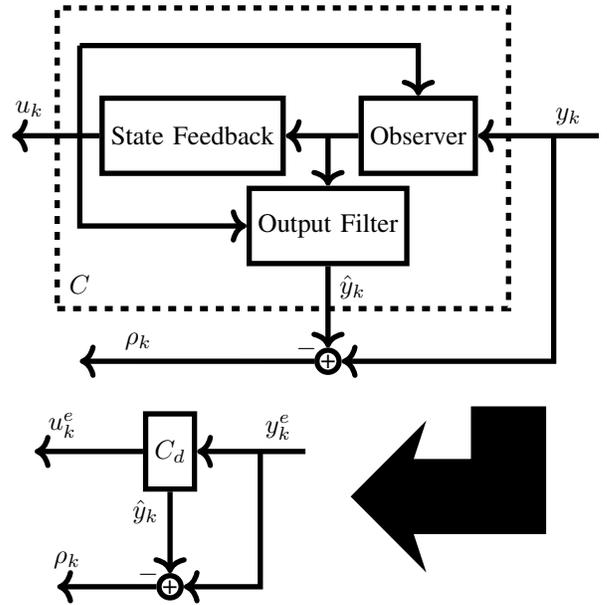
\begin{figure}
     \centering
    \begin{tikzpicture}[scale=0.6, every line/.style={scale=0.6}]
    \node[draw, line width = 2, rectangle, minimum height=1cm] (state) at (-3,0) {State Feedback};
    \node[draw, line width = 2, rectangle, minimum height=1cm] (observer) at (2,0) {Observer};
    \node[draw, line width = 2, rectangle, minimum height=1cm] (delay) at (0,-2) {Output Filter};
    \node[draw, line width = 2, circle, inner sep=0.01cm, minimum width=0.1cm, minimum height=0.1cm] (plus) at (0,-5) {+};
    
    \draw[line width=2, ->] (6,0) -- node[pos=0.25,above] {$y_k$} (observer);
    \draw[line width=2, ->] (observer) -- (state);
    \draw[line width=2, ->] (state) -- node[yshift=10,xshift=-10] {$u_k$} (-7,0);
    \draw[line width=2, -] (-5.5,0) -- (-5.5,2);
    \draw[line width=2, ->] (-5.5,2) -| (observer);
    \draw[line width=2, ->] (0,0) -- (delay);
    \draw[line width=2, ->] (delay) -- node[pos=0.85, left,yshift=-5] {$-$} node[pos=0.25,right] {$\hat y_k$} (plus);
    \draw[line width=2, ->] (5,0) |- (plus);
    \draw[line width=2, ->] (-5.5,0) |- (delay);
    \draw[line width=2, ->]  (plus) -- node[pos=0.75, above] {$\rho_k$} (-5.5,-5);

    \node[draw, line width = 2,rectangle,fill=none,dashed, minimum width=6cm, minimum height=4cm] at (-1,-0.5) {};

    \node at (-5.5,-3.3) {$C$};

    \draw[line width=2, -{Stealth[length=10mm, width=20mm]}, black ,line width=10mm] (4,-6) |- (0.5,-8);
    
    \begin{scope}[xshift=-3.5cm,yshift=-7cm]
        \node[draw, line width = 2, rectangle, minimum height=1cm] (state) at (0,0) {$C_d$};
        \node[draw, line width = 2, circle, inner sep=0.01cm, minimum width=0.1cm, minimum height=0.1cm] (plus) at (0,-3) {+};
        
        \draw[line width=2, ->] (state) -- node[yshift=10,xshift=-10] {$u^e_k$} (-3,0);
        \draw[line width=2, ->] (3,0) -- node[pos=0.25,above] {$y_k^e$} (state);
        \draw[line width=2, ->] (state) -- node[pos=0.85, left,yshift=-5] {$-$} node[pos=0.25,left] {$\hat y_k$} (plus);
        \draw[line width=2, ->]  (plus) -- node[yshift=10,xshift=-15] {$\rho_k$} (-2.5,-3);
        \draw[line width=2, ->] (2,0) |- (plus);
    \end{scope}
\end{tikzpicture}
    \caption{Instead of individually converting each block in the controller to be compatible with encrypted signals, we reduce the compounding approximation errors by combining the three blocks (two static and one dynamic) into a single dynamic block, $C_d$ with two outputs. It is made compatible with integer signals by scaling and rounding its parameters.
    }    \label{fig:controller_enc_fig}
\end{figure}
Furthermore, as we saw in Lemma~\ref{lem:ishomomorphic}, additions and multiplications increase the variance of the injected noise during encryption. Implementing controllers over encrypted signals should, therefore, ideally, minimize the number of computational steps to avoid a potential detour into numbers below the precision limit or deterioration of the result due to the increase in noise. Controllers and detectors in the traditional encrypted-free setting contain internal loops and are split up into multiple steps, which is shown in the top part of Fig.~\ref{fig:controller_enc_fig}. For the encrypted setting, it is beneficial if as few operations as possible are done. In this paper, we will aggregate the operations into a single transformation, depicted in the bottom part of Fig.~\ref{fig:controller_enc_fig} as $C_d$. However, $C_d$ will still have an internal state variable, meaning that internal loops within it are not fully removed.
The estimated output of the plant $y^e_k$ based on its internal state $x_k^c$ and a sequence of possibly corrupted past measurements can be written as
\begin{equation}\label{eq:controller_estimates_output}
         \hat{y}^e_k = C_ex_k^c+D_e\begin{bmatrix}
        \tilde y_0^e \\ \vdots \\ \tilde y_{k-1}^e
    \end{bmatrix}
    \end{equation}
where $C_e$ and $D_e$ are matrices with integer elements.

Equation~\ref{eq:controller_estimates_output} does not limit control performance since the extra signals that the controller outputs do not affect the encrypted input signals, $u_k^e$. Thus, as long as the controller can be cast in a compatible integer formulation, we can add an extra output using the computation it already makes.

The filtered signal $\hat y^e_k$ in Equation~\ref{eq:controller_estimates_output} will also be scaled dynamically with another quantization factor, which we will see in Section~\ref{sec:controllers}, related to the \emph{multiplication depth} of the cryptosystem. The estimated output $\hat y^e_k$ must match the scaling of $\tilde y^e_k$ to create a residual. For now, we shall say this scaling is done implicitly. 
The anomaly detector looks, therefore, at the following encrypted residual,
\begin{equation}\label{eq:encrypted_residual} 
    \rho_k^e = (\tilde y^e_k-\hat{y}_k^e) \bmod q,
\end{equation}
where $\tilde y_k^e =  y^e_k+y^a_k$, to detect an attack. 
We model an attacker that can additively inject signals $u^a_k$ and $y^a_k$ onto the actuators and sensors, respectively. During a disclosure phase, the attacker can also read the signals $y^e_k$ and $u^e_k$ and generate attacks based on them using an algorithm the defender does not know about. It is possible that an attacker can learn \emph{undetectable} attacks from the data~\cite{taheri2021,alisic2021ecc}, which implies that $\rho_k=0, \, \forall k$. See Definition 1 and Lemma 3.1 in~\cite{Pasqualetti2013} for a definition of undetectable attacks. However, to make our analysis tractable, we require that the attacks can be detected.
\begin{asm}\label{asm:detectable}
The attacker injects detectable attacks. Specifically, the attacks occasionally make $\rho_k \neq 0$.
\end{asm}
Solving the following problem enables detection.
\begin{prob}\label{prob:main}
    Consider the following two hypotheses:
    \begin{itemize}
        \item $\mathcal{H}_0:$ $\dec^\lambda_s( \rho_k^e)=0$, $\forall k$,
        \item $\mathcal{H}_1:$ $\dec^\lambda_s( \rho_k^e)\neq 0$ for some time steps $k$.
    \end{itemize}
    Without using the secret vector $s$, how can $\mathcal{H}_0$ be rejected in favor of $\mathcal{H}_1$ with a specific probability $\alpha$ on the Type~I Error?
\end{prob}

In plain detection terminology, Problem~\ref{prob:main} asks how to detect anomalies in encrypted signals with a set false alarm rate. We decouple the problem of detection \emph{given} a residual signal, from the control-theoretic aspects of obtaining the residual. Furthermore, the residual signal is rarely always zero as in $\mathcal{H}_0$. Even in a deterministic setting, small and sparse errors are introduced due to the quantization of the signals. An example would be limit cycles for the control of unstable plants. We shall therefore also seek to detect attacks under such scenarios:
\begin{prob}\label{prob:secondary}
    Consider the following two hypotheses:
    \begin{itemize}
        \item $\mathcal{H}_0^e:$ $\Vert \dec^\lambda_s( \rho^e)\Vert _\infty \leq \epsilon_\infty$ and $\Vert \dec^\lambda_s( \rho^e)\Vert_0 \leq \epsilon_0$,
        \item $\mathcal{H}_1^e:$ $\Vert \dec^\lambda_s( \rho^e)\Vert _\infty > \epsilon_\infty$ or  $\Vert \dec^\lambda_s( \rho^e)\Vert _0 > \epsilon_0$,
    \end{itemize}
    where $\rho^e= \begin{bmatrix}\left(\rho_1^e\right )^\top & \dots & \left(\rho_N^e \right)^\top \end{bmatrix}^\top$. Without using the secret vector $s$, how can $\mathcal{H}^e_0$ be rejected in favor of $\mathcal{H}^e_1$ with a false alarm probability $\alpha$?
\end{prob}

At first glance, Problem~\ref{prob:secondary} seems to be a mere attempt at error-handling for potentially sparse errors. However, it should be viewed as the \emph{controller's impact on detection strength}, a relatively novel idea in detection literature. By proposing a mechanism that solves the hypothesis tests, we can determine which characteristics of the residual are necessary for detection. Furthermore, we will also observe how these characteristics affect detection strength, which can be used to design controllers. However, we will save the design problem for future work.
The rest of the paper is organized to solve these two sub-problems:
\begin{itemize}
    \item \label{point:detect} \emph{Detection}: We develop a theory that allows us to perform a hypothesis test on encrypted signals and solve Problems~\ref{prob:main} and~\ref{prob:secondary} given access to the residual \eqref{eq:encrypted_residual}.
    \item \emph{Characteristics}: We show what characteristics from the residual must be known to perform detection and how those characteristics affect the detection strength.
\end{itemize}
The detection problem will partially utilize tools found in the cryptography community, such as searching for short vectors in a lattice. In contrast, we will use control-theoretic tools to obtain the characteristics of the residual signal.

\section{Main Results}\label{sec:results}

We shall now consider the first method we need for detection. It was first reported in~\cite{alisic2023modelfreelwe} as part of a method to bypass the encryption when learning attacks. In particular, consider the encrypted residual messages from~\eqref{eq:encrypted_residual}:
\begin{equation*} 
    \enc^\lambda_s(\rho_k) = (P_k^\rho, P_k^\rho s+r\rho_k+e_k^\rho) \bmod q.
\end{equation*}
where we have redefined $P_k^\rho= (P_k^y-P_k^{\hat y}) $ and $\rho_k= (\tilde y_k-\hat y_k)$ for brevity. We will start by ``filtering out" the part that depends on the secret vector, $P_k^\rho s$, by finding linear combinations of encrypted messages that remove it.

Finding such a \emph{filtering vector} $d$ can be done algorithmically by computing the corresponding $q$-ary lattice for the matrix comprised of vectorized matrices $P_k^\rho$ stacked as columns in
\begin{equation} \label{eq:matrix_public_vectorized}
    \mathcal{P}^\rho= \begin{bmatrix}
        \vecc P_0^\rho & \vecc P_1^\rho &\cdots & \vecc P_{N-1}^\rho \end{bmatrix}.
\end{equation}
The filtering vectors $d$ define the points that belong to the $q$-ary lattice associated with~\eqref{eq:matrix_public_vectorized}, namely, $d \in \Ker_q {\mathcal{P}^\rho}$. 
There will typically be $N-pv$ such vectors when $N$ is large. 

Multiplying such a filtering vector $d=\begin{bmatrix} d_1 & d_2 & \dots \end{bmatrix}^\top$ with the residuals $\rho^e$, modulo $q$, gives
\begin{equation} \label{eq:filtered_ensemble}
    \begin{aligned}
      d^\top \rho^e \bmod q  &= (0, \sum \limits_k (\rho _kr+e_k^\rho)d_k \bmod q) \\
    & = (0, r d^\top \rho + d^\top e^\rho) \bmod q \\& = (0, r d^\top \rho + d^\top e^\rho),
    \end{aligned}
\end{equation}
where we have assumed for the last equality that $\vert r d^\top \rho + d^\top e^\rho \vert \leq \frac{q}{2}$. 
The noise term in~\eqref{eq:filtered_ensemble}, $d^\top e^\rho$, obscures the combined plaintext sum of messages, $r d^\top \rho$, thus ensuring the confidentiality of individual samples. We can state the following proposition thanks to Assumption~\ref{asm:noise_var}.
\begin{prop} \label{prop:error_dist}
The error vector $e^\rho$ from the residuals is statistically close to the zero-mean discrete Gaussian, $e^\rho \sim \mathcal{N}_{\mathbb{Z}_q}(0,\Sigma_\rho)$, for some covariance matrix $\Sigma_\rho$.
\end{prop}
\begin{proof}
    The result follows from choosing a level of the noise so that Assumption~\ref{asm:noise_var} holds, Lemma~\ref{lem:ishomomorphic}, and that we have only considered linear combinations when computing the residual in~\eqref{eq:controller_estimates_output}, which we for now denote by the matrix $T_c$.
\end{proof}

The resulting covariance matrix $\Sigma_p$, depends on the operations done to the signal due to the controller, which we write as $\Sigma_p=\sigma^2T_c^\top T_c$. It's precise formulation will be explicitly derived in Section~\ref{sec:controllers}. Knowing $T_c$ is important here, as it is integral to the detection algorithm. In particular, Proposition~\ref{prop:error_dist} leads us to the following result.
\begin{cor}~\label{cor:distribution_residual_weighted}
    The non-zero message in~\eqref{eq:filtered_ensemble} is distributed as a discrete normal, $d^\top\rho^e \bmod q =(0,X)$, where:
    \begin{equation}\label{eq:normal}
         X \sim \mathcal{N}_{\mathbb{Z}_q}(rd^\top \rho, d^\top \Sigma_p d),
    \end{equation}
    within a statistical distance $\mathcal{O}(\epsilon)$ if $\Vert T_c^\top d\Vert_2^2 \eta_\epsilon^2(\mathbb L) \leq \sigma^2$.
\end{cor}

Corollary~\ref{cor:distribution_residual_weighted} shows us how to perform the hypothesis test we posed in Problem~\ref{prob:main}, where we assumed the ideal case of a plaintext residual being zero at all times:
\begin{thm}\label{thm:main_theorem_rejection_rule}
    Rejecting $\mathcal{H}_0$ when $ \vert X \vert \geq \gamma$, where $d^\top\rho^e \bmod q =(0,X)$ and $\gamma$ is the $\alpha$-quantile of $\mathcal{N}_{\mathbb{Z}_q}(0,d^\top \Sigma_p d)$,
    \begin{equation}  \label{eq:false_alarm_triggering_rule}
    \mathrm{Pr}(\vert Y \vert\geq \gamma)=\alpha, \quad \text{for } Y \sim \mathcal{N}_{\mathbb{Z}_q}(0,d^\top \Sigma_p d),
    \end{equation}
    produces a Type~I Error with probability $\alpha$ with a small statistical distance of $\mathcal{O}(\epsilon)$ if $\Vert T_c^\top d\Vert_2^2 \eta_\epsilon^2(\mathbb L) \leq \sigma^2$.
\end{thm}

\begin{proof}
Recall that the Type~I Error is defined as $\mathrm{Pr} (\text{Reject } \mathcal{H}_0 | \mathcal{H}_0 \text{ is true})=\alpha$. Thus, we have to construct a rejection mechanism with this property. Hypothesis $\mathcal{H}_0$ implies that $d^\top \rho=0$, which according to Corollary~\ref{cor:distribution_residual_weighted}, implies that:
    \begin{equation*} 
        X \sim \mathcal{N}_{\mathbb{Z}_q}(0, d^\top \Sigma_p d).
    \end{equation*}
    Consider the random variable $Y\sim \mathcal{N}_{\mathbb{Z}_q}(0, d^\top \Sigma_p d)$ and the sets $\mathcal{S}_i$, for which $\mathrm{Pr}(Y\subset \mathcal{S}_i)=\alpha$. Multiple such sets may exist, but we choose to use $\mathcal{S}_0={Y: Y<-\gamma, \text{ or }  Y > \gamma}$, for some $\gamma\in \mathbb{Z}_q$. Thus, we have that $\mathrm{Pr}(|Y|\geq \gamma)=\alpha$. Since $d^\top\rho^e \bmod q=X$ has the same distribution as $Y$ under $\mathcal{H}_0$, we have that: $$\mathrm{Pr}\left (\left. |X|\geq \gamma \right \vert \mathcal{H}_0 \text{ is true} \right )=\alpha,$$ which concludes our proof.
\end{proof}
\begin{rem}
    Several rejection rules that solve Problem~\ref{prob:main} appear in the proof of Theorem~\ref{thm:main_theorem_rejection_rule}, quantified by the number of sets $\mathcal{S}_i$ we can choose from. These sets are equivalent to how we could pick different rejection regions in Section~\ref{sec:hyptest}. The one we use in Theorem~\ref{thm:main_theorem_rejection_rule} is the \emph{uniformly most powerful unbiased} $\alpha$-level test~\cite{Lehmann2005}.
\end{rem}

To extend the detection mechanism to solve Problem~\ref{prob:secondary}, we return to the last line of~\eqref{eq:filtered_ensemble}, where the plaintext residual is given by $rd^\top \rho$. Since the detector does not know $\rho$, we take a worst-case approach to $\rho$'s values, even if the vector is sparse. We shall therefore approximate:
\begin{equation*}
    rd^\top \rho \approx r Y, \text{ where } Y\sim\mathcal{U}_{\epsilon_\infty \Vert d_{\epsilon_0} \Vert_1},
\end{equation*}
where $d_{\epsilon_0}$ is a vector consisting of $d$'s $\epsilon_0$ largest values. The motivation is as follows, if it is known that $\Vert \rho \Vert_0 \leq \epsilon_0$, then the largest values $d^\top \rho$ can take is $\pm \epsilon_\infty \Vert d_{\epsilon_0} \Vert_1$. In case there is no sparsity, $\epsilon_0 = Np$, then $d_{\epsilon_0}=d$. With this approximation of the error, we have:
\begin{prop}\label{prop:prop_rejection_rule}
    Rejecting $\mathcal{H}^e_0$ when $  X \in \mathcal{S}$, where $d^\top\rho^e \bmod q =(0,X)$ and $\mathcal{S}$ is a set where $\mathrm{Pr}(Y \in \mathcal{S})=\alpha$ for the random variable $Y$ with the following density function
    \begin{equation*}
        \mathrm{Pr}(y=Y) = \frac{1}{1+2 \epsilon_\infty \Vert d_{\epsilon_0} \Vert_1} \sum \limits_{k=-\epsilon_\infty \Vert d_{\epsilon_0} \Vert_1}^{\epsilon_\infty \Vert d_{\epsilon_0} \Vert_1} \mathrm{Pr}(y-kr=X),
    \end{equation*}
    where $\mathrm{Pr}(x=X)$ is the density function for $\mathcal{N}_{\mathbb{Z}_q}(0,d^\top \Sigma_p d)$, produces a Type~I Error with probability $\alpha$ with a small statistical error of $\mathcal{O}(\epsilon)$.
\end{prop}
\begin{proof}
    Since $r d^\top \rho + d^\top e^\rho$ is a sum of a random variable sampled from a discrete normal distribution (up to $\mathcal{O}(\epsilon)$ statistical distance) and a random variable that we approximate as uniformly distributed, $\mathcal{U}_{\epsilon_\infty \Vert d_{\epsilon_0} \Vert_1}$, One can obtain the resulting density function by computing the convolution between the two density functions. 
\end{proof}
Similarly to the proof in Theorem~\ref{thm:main_theorem_rejection_rule}, multiple sets $\mathcal{S}$ may exist. One obtains the \emph{uniformly most powerful} $\alpha$-level test by picking $\mathcal{S}$ with the largest cardinality. While the largest set in the error-free case consists of points far away from the origin, the same may not hold for the case when there are errors.

Theorem~\ref{thm:main_theorem_rejection_rule} and Proposition~\ref{prob:secondary} gives rules for performing a hypothesis test once some vector $d$ has been found. 
While the hypothesis test depends on obtaining good filtering vectors $d$ on the cryptographic side of the problem, $d$ essentially only scales the difficulty of performing the test and a new $d$ will be obtained once a new sample is introduced into the problem. In particular, the error variance $d^\top \Sigma_\rho d$ of the discrete normal variance is scaled quadratically with $d$, and its $\epsilon_0$ largest values linearly scale the impact of quantization errors. In contrast, the system parameters $\Sigma_p$, $\epsilon_\infty$ , and $\epsilon_0$ from the control-theoretic side determine the system's best-case detection capabilities.

We also need to recognize the temporal nature of the problem. It is most valuable if anomaly is detected in the \emph{most recent samples}. The samples that are ignored are encoded in the filtering $d$ through its elements that are zero. Therefore, we want to ensure that the most recent samples are used in the detection scheme. Therefore, we want the filtering vector $d$ to have the following two properties:
\begin{enumerate}
    \item Structured: The first elements in $d$ should be 0 so that only recent samples are used for detection.
    \item Short: The norms $\Vert d_{\epsilon_0} \Vert_1$ and $\Vert d\Vert_2$ should be small to reduce error amplification.
\end{enumerate}
In Section~\ref{sec:vectors}, we show that these two properties can oppose each other and that a trade-off between them has to be considered.


\section{Hypothesis Testing under LWE}\label{sec:hyptest}
The proposed detection scheme seems, at first glance, like it violates the security properties of LWE, in particular since we still are dependent on finding a short vector. In this section, we will clarify the difference between the two problems, and show that the problem related to detection actually scales polynomially with the length of the shortest vector, as opposed to exponentially. Hypothesis tests are used to exemplify the hardness of the LWE problem in the original work by Regev~\cite{regev2009}. In particular, it is shown that if an efficient algorithm can reject that a set of samples comes from the uniform distribution in favor of a so-called \emph{LWE distribution}, see Definition~\ref{def:lwedist} below, then the secret vector $s$ can also be recovered efficiently. However, the existence of such a rejection algorithm would imply that the LWE problem is easy to solve. Let us now define what we mean by the LWE distribution.

The uniform distribution $\mathcal{U}_q$ can be combined with the discrete normal distribution from Definition~\ref{def:discretenormal} to make the LWE distribution as follows.
\begin{defi} \label{def:lwedist}
A sample from the LWE distribution $\mathcal{L}_{v,q,\sigma^2}(s)$ with parameters $v$, $q$, and $\sigma^2$ is defined as a tuple, $(P,b)$, where the elements of $P \in \mathbb Z_q^{v}$ are sampled from $\mathcal{U}_q$, and $b$ is given by,
\begin{equation*}
    b \coloneqq Ps+e \bmod q,
\end{equation*}
where $e \sim \mathcal{N}_{\mathbb{Z}_q}(0,\sigma^2)$. The vector $s \in \mathbb Z_q^v$ parametrizes the specific LWE-distribution.
\end{defi}

Comparing to the encryption scheme in Definition~\ref{def:lwe}, and specifically~\eqref{eq:lweencmsg}, one can see that a sample from the LWE distribution is the same as encrypting the plaintext message $m=0$. In particular, we can identify the null hypothesis $\mathcal{H}_0$ in Problem~\ref{prob:main} as being equivalent to having a collection of samples from the LWE distribution,
\begin{itemize}
    \item $\mathcal{H}_0$: $\rho_k^e\sim \mathcal{L}_{v,q,\sigma^2}(s)$, $\forall k$.
\end{itemize}

The test we used in Section~\ref{sec:results} finds integer-linear combinations of messages so that the result, after modulo $q$, removes the explicit dependence on $P_k^\rho$,
\begin{equation*}
    d^\top \rho^e \bmod q = (0, d^\top e \bmod q).
\end{equation*}
Now, if the result is sufficiently small, namely if $\vert d^\top e \vert < \frac{q}{2}$, we will have an integer-linear combination of samples distributed as discrete normal which, thanks to Assumption~\ref{asm:noise_var}, is also distributed ($\epsilon$-close) as a discrete normal, see~\cite{boneh2011}. Therefore, we can convert the null hypothesis to:
\begin{itemize}
    \item $\mathcal{H}_0$: $d^\top \rho^e \bmod q = (0,X)$, where $X\sim \mathcal{N}_{\mathbb{Z}_q}(0,\Vert d\Vert_2^2\sigma^2)$. 
\end{itemize}
In other words, the null hypothesis is that $X$ is sampled from the discrete normal distribution.

Consider now the alternative hypothesis $\mathcal{H}_1$, which assumes that (at least some of) the plaintext residual samples are non-zero. Performing the same integer-linear combination of samples yields
\begin{equation}\label{eq:uniform_result}
    d^\top \rho^e \bmod q = (0, d^\top \rho r + d^\top e \bmod q),
\end{equation}
where $\rho$ is some strictly non-zero vector. The effect of the non-zero $\rho r$, especially if it is sufficiently large, can be understood by comparing the second tuple in~\eqref{eq:uniform_result} with the following class of pseudo-random uniform number generators, see~\cite{rotenberg60}, which outputs sequences that pass formal test of randomness, however, can be predicted if the parameters are known.

\begin{lem}\label{lem:makeUnif}
    Consider the two integers $a\in\mathbb Z_q$ and $c \in \mathbb Z_q$ for prime $q$ and let $X \sim \mathcal{U}_q$. Then, the following affine transformation:
    \begin{equation} \label{eq:uniformize}
        Y= aX+c \bmod q
    \end{equation}
    is also pseudo-uniform, $Y \sim \mathcal{U}_{q}$, if $a$ and $q$ are coprime.
\end{lem}

Essentially, the information about $a$ and $c$ is destroyed because they cannot be recovered from $Y$ even if the realization of $X$ is known. We can, therefore, identify the second tuple in~\eqref{eq:uniform_result} as a sum of pseudo-uniform random number generators with different parameters $a_k=\rho_kr$ and $c_k=d_ke_k$,
\begin{equation*}
   d^\top \rho r + d^\top e \bmod q =  \left ( \sum \limits_k d_k\rho_kr + d_ke_k \bmod q \right) \bmod q.
\end{equation*}
The final modulo operation maps the sum of pseudo-uniform random numbers over $\mathbb Z_q$ back to $\mathbb Z_q$, which follows a pseudo-uniform distribution again, as is shown next.
\begin{lem}
\label{lem:sum_pseudo_unif}    Consider a set of messages $u_k$ sampled from the uniform distribution, $u_k \sim \mathcal{U}_{q}$ and let $q$ be prime. Then, the linear combination
    \begin{equation*}
        U \coloneqq \sum \limits_k a_k u_k \bmod q,
    \end{equation*}
    where $a_k \in \mathbb Z_q$ and $a_k \neq 0$ for some $k$, is also uniform, $U \sim  \mathcal{U}_{q}$.
\end{lem}
\begin{proof}
    First, note that the sum of two uniformly distributed variables $u_i, u_j \sim  \mathcal{U}_{q}$, is also uniform after the modulus operator has been applied,
    \begin{equation*}
        u_i+u_j \bmod q \sim  \mathcal{U}_{q}, \quad \text{ if } i\neq j.
    \end{equation*}
    Secondly, multiplying a uniform number $u_k$ with an integer $a_k \in \mathbb Z_q$ is also integer,
    \begin{equation*}
        a_k u_k \bmod q \sim  \mathcal{U}_{q},
    \end{equation*}
    if $a_k \neq 0$, since $a_k$ and $q$ are coprime. The last part of the proof follows from recurrence. Let us define
    \begin{equation*}
        S_{k+1} = S_k + a_{k+1}u_{k+1} \bmod q, \quad S_1 = a_1u_1 \bmod q.
    \end{equation*}
Note that since $S_k$ and $a_{k+1}u_{k+1} \bmod q$ are both uniform, then their sum modulo $q$, $S_{k+1}$, is uniform too.
\end{proof}


The uniformity of~\eqref{eq:uniform_result} allows us to reformulate Problem~\ref{prob:main} into the following one, which is the problem we solve when performing anomaly detection.
\setcounter{prob}{1}
\begin{subprob} \label{prob:reworked_prob}
    Consider the following two hypotheses:
    \begin{itemize}
        \item $\mathcal{H}_0$: $d^\top \rho^e \bmod q = (0,X)$, where $X\sim \mathcal{N}_{\mathbb{Z}_q}(0,\Vert d\Vert_2^2\sigma^2)$. 
        \item $\mathcal{H}_1$: $d^\top \rho^e \bmod q = (0,U)$, where $U\sim \mathcal{U}_q$. 
    \end{itemize}
    Reject $\mathcal{H}_0$ in favor of $\mathcal{H}_1$ with a specific probability $\alpha$ on the Type~I Error.
\end{subprob}


In the following, we shall explain the subtle difference between Problem~\ref{prob:reworked_prob} and the one treated in~\cite{regev2009}, which reveals the secret vector. 
Consider a collection of samples from the LWE distribution, denoted by $(P_k, P_k s + e_k \bmod q)$, where $k$ is a label for each sample. Then, draw an integer from the uniform distribution for each sample, $p_k \sim \mathcal{U}_{q}$ and modify each sample as follows.
\begin{align*}
    & (P_k+p_k \begin{bmatrix}
        1 & 0 & \dots & 0 
    \end{bmatrix} \bmod q, P_k s + e_k \bmod q) \\ 
    = & (\tilde P_k, P_ks +e_k +p_ks_1 -p_ks_1 \bmod q) \\ 
    = &(\tilde P_k, \underbrace{-s_1}_a \underbrace{p_k}_X + \underbrace{\tilde P_ks +e_k}_{c}  \bmod q).
\end{align*}

The last element in the transformed tuple sample can be compared to the output of the pseudo-random uniform number generator~\eqref{eq:uniformize}, where we can identify the known $p_k$ as $X$, and the unknown parameters $-s_1$ and $\bar P_k + e_k$ as $a$ and $c$, respectively. In other words, we have \emph{transformed our LWE samples into} (pseudo-random) \emph{uniform samples}.


Uncovering the first element of secret vector $s$, namely $s_1$, occurs when a transformation back to the LWE distribution is performed \emph{and verified}. In particular, first make a guess of $s_1$, which we denote as $\tilde s_1$. Then, multiply it with the known integer $p_k$, and add the product to the second part of the tuple
\begin{equation*}
    (\tilde P_k, \underbrace{(\tilde s_1 - s_1)}_a \underbrace{p_k}_X + \underbrace{\tilde P_ks +e_k}_{c}  \bmod q).
\end{equation*}
Note that the sample remains uniform if a wrong guess is made, $\tilde s_1 \neq s_1$. The only thing that changes is the (unknown) parameter $a$ in the pseudo-random number generator. However, if a correct guess is made, $\tilde s_1=s_1$, the samples will no longer be uniform. Rather, they will be LWE samples again. If there is a way to verify that the samples have been transformed back to LWE, then the (first element of the) secret vector has been revealed. Therefore, to reveal elements of the secret vector, the following problem has to be solved.
\begin{regevProblem}
    Consider the two hypotheses in Problem~\ref{prob:reworked_prob}. Reject $\mathcal{H}_1$ (Uniform) in favor of $\mathcal{H}_0$ (LWE) with a specific probability $\alpha$ on the Type~I Error.
\end{regevProblem}

If $\mathcal{H}_1$ can be rejected efficiently (polynomial complexity), then at most $q$ guesses per element of the secret vector are needed, $\mathcal{O}(qv)$, making the entire scheme polynomial.



We summarize the two different problems and their implications in Fig.~\ref{fig:visual_break_LWE}. It is the verification in the second step that reveals the secret vector in~\cite{regev2009}. When we solve Problem~\ref{prob:reworked_prob} to perform anomaly detection, we verify the first step in Fig.~\ref{fig:visual_break_LWE}, which would \emph{not reveal the secret vector}. The rest of this section is devoted to showing that solving Problem~\ref{prob:reworked_prob} can be done more efficiently.

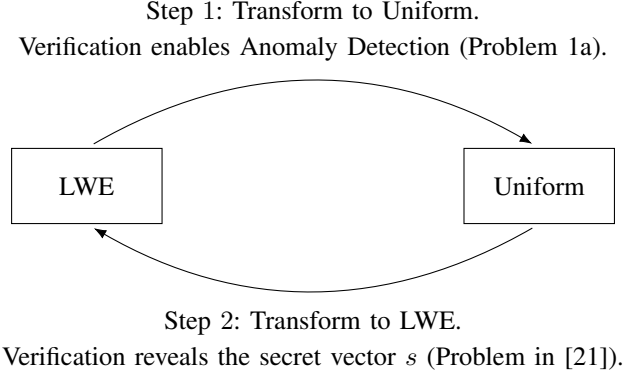
\begin{figure}
    \centering
\begin{tikzpicture}[
  node distance=2cm and 4cm,
  box/.style={draw,minimum width=2cm,minimum height=1cm,align=center},
  arc/.style={-Latex,shorten >=3pt,shorten <=3pt}
  ]
  \node[box] (LWE) {LWE};
  \node[box,right=of LWE] (Uniform) {Uniform};
  \draw[arc] (LWE.north) to[bend left] node[midway, above, yshift=4ex] {Step $1$: Transform to Uniform.} node[midway, above, yshift=1ex]{Verification enables Anomaly Detection (Problem~\ref{prob:reworked_prob}).} (Uniform.north);
\draw[arc] (Uniform.south) to[bend left] node[midway, below, yshift=-1ex] {Step $2$: Transform to LWE.} node[midway, below, yshift=-4ex] {Verification reveals the secret vector $s$ (Problem in~\cite{regev2009}).}  (LWE.south);
\end{tikzpicture}    
    \caption{The schematic shows a high-level picture of how the secret vector is revealed using the algorithm laid out in~\cite{regev2009}. First, the LWE sampled is transformed into a uniform one by modifying the public key, as shown in Step 1 (confirmation that this transformation is successful is not treated in~\cite{regev2009}). Then, by guessing the correct value of the secret vector, one can return the sample back to an LWE sample, as in Step 2. The ability to confirm that the sample is LWE again verifies the guess and thus reveals the secret vector.}
    \label{fig:visual_break_LWE}
\end{figure}

Our scheme transforms the samples into a single, yet equivalent, sample from either a uniform or a discrete normal distribution (if the integer-linear combination $d$ gives a sufficiently small output). In Fig.~\ref{fig:unif_norm}, the uniform and discrete normal distributions are shown. The typical way to design a hypothesis test is to settle for a level of the Type~I Error $\alpha$, the false alarm probability. Then, one chooses a region that minimizes the Type~II Error for the alternative hypothesis and whose probability under the null hypothesis is, at most, the false alarm probability. Typical choices of these regions are shown in Fig.~\ref{fig:unif_norm} for a Type~I Error probability of $5 \%$ when the message space is $q=227$.

\begin{rem}
    Note that Lemmas~\ref{lem:makeUnif}
 and~\ref{lem:sum_pseudo_unif} lend an interpretation to Proposition~\ref{prop:prop_rejection_rule} that solves Problem~\ref{prob:secondary}. The quantity $rd^\top \rho + d^\top e^\rho \bmod q$ for small and sparse $\rho$ becomes a sum of pseudo-uniform random number generators with poorly chosen parameters. An example of poor parameters for~\eqref{eq:uniformize} is $a=1$, $c=1$, and $q\gg1$. Given a set of $Y$ and $X$, one could identify $a$ and $c$ because of the tiny increments.
 \end{rem}
\begin{figure}
    \centering
    \includegraphics[scale=0.45]{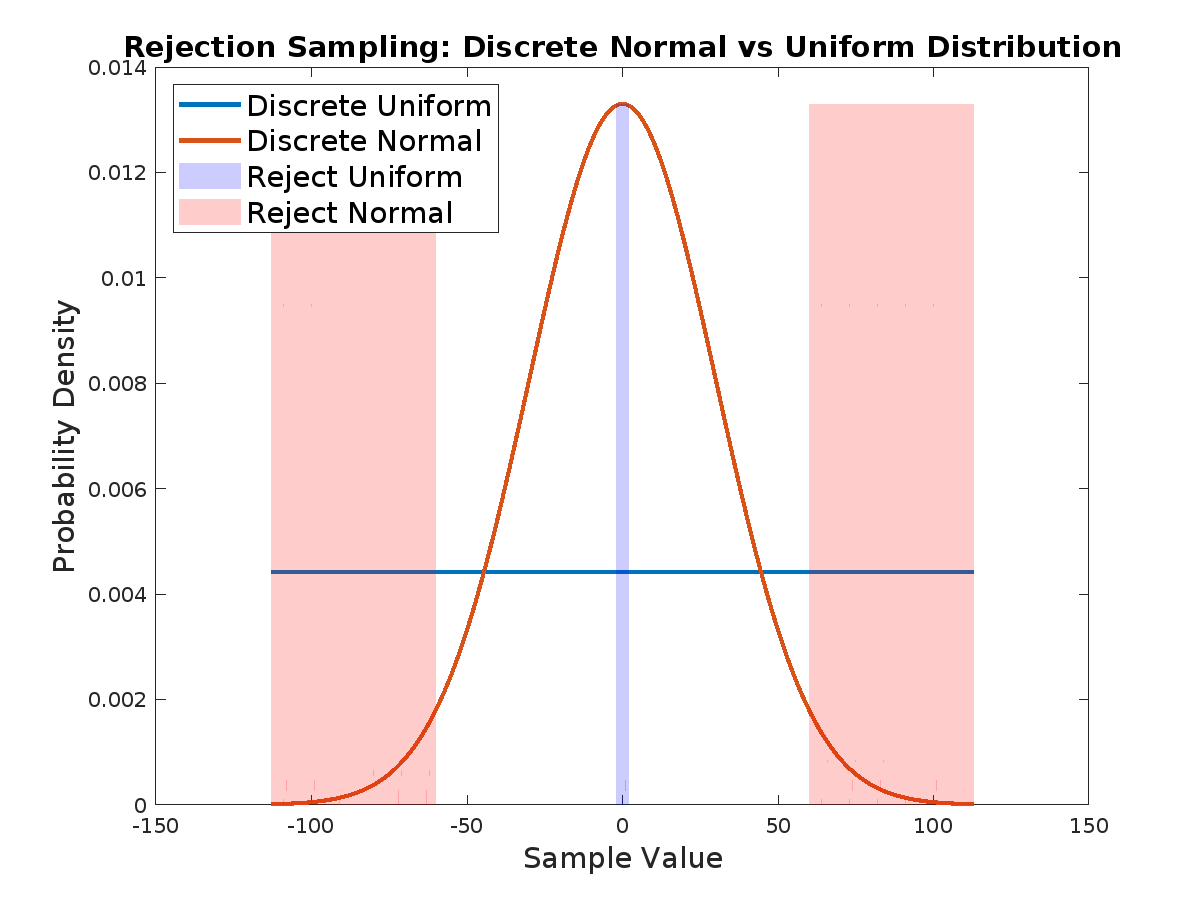}
    \caption{The graph shows the probability distribution for the discrete normal and uniform distributions. Note that rejecting the uniform distribution results in a very small rejection region, whose maximal power is achieved by centering it around the mean of the alternative distribution, which is the discrete normal here. If we design a test to reject the other hypothesis, one may see that the resulting rejection region in red, to reject the discrete normal, is much larger. Specifically, its false-positive probability is minimized by placing the rejection region around the tails.}
    \label{fig:unif_norm}
\end{figure}

When sampling from a uniform distribution, the Type~I Error probability, $\alpha$, remains constant regardless of the chosen rejection region, as long as it contains fewer than $N<\alpha q$ points. However, the Type~II Error is minimized when the rejection region is centered around the mean of the alternative hypothesis, meeting the criteria for a Neyman-Pearson test, thus becoming the uniformly most powerful test~\cite{Lehmann2005}. Despite this, the Type~II Error grows exponentially with the variance and $\Vert d\Vert_2^2$, as indicated by the red curve in Fig.~\ref{fig:t2}. Conversely, when rejecting a discrete normal distribution, the Type~II Error is minimized by maximizing the rejection region's size by concentrating the region on the distribution's tails. However, the rejection region dynamically shrinks as the variance increases, resulting in slower growth of the Type~II Error, as shown by the blue step-curve in Fig.~\ref{fig:t2}. This suggests that rejecting the normal distribution in favor of the uniform distribution is easier than vice versa, which we prove next. 

\begin{figure}
\hspace{-20pt}
\includegraphics[width=0.55\textwidth,trim={0 50 0 50}, clip]{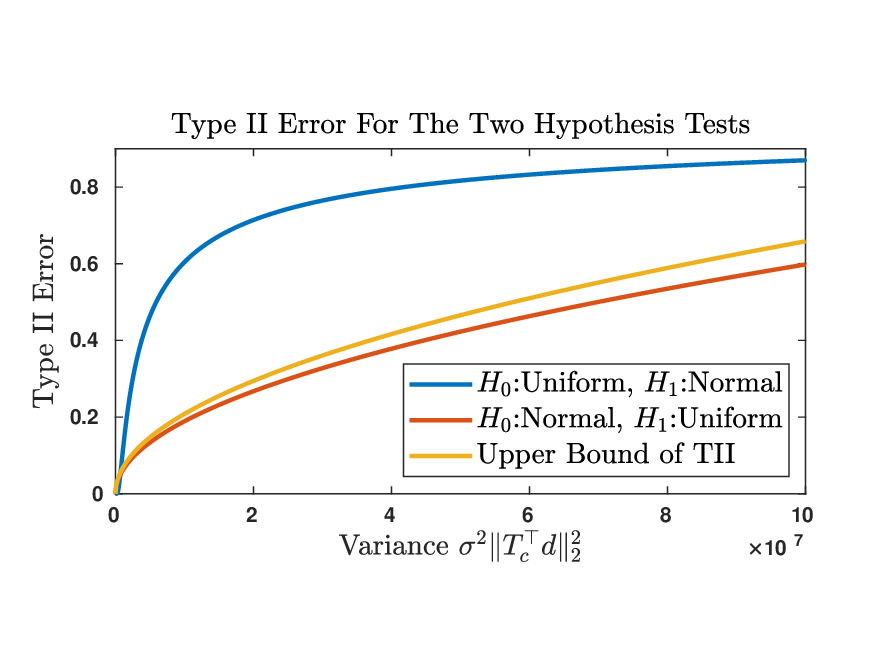}
    \caption{The graph shows the Type II Error of the two hypothesis tests as a function of the resulting variance using $q=10^{16}$ (256-bit security if $v=1024$ and $\sigma^2=10$). Note the exponential convergence of rejecting the uniform distribution (red curve). It indicates that for the test to be statistically significant, an (exponentially) short vector must be found to solve the underlying lattice problem. For the other hypothesis test, used in our anomaly detector, the convergence rate is not exponential, which is proven by the upper bound from Theorem~\ref{thm:bounds}, given by the yellow curve.}
    \label{fig:t2}
\end{figure}


\begin{thm}\label{thm:bounds}
    The Type II Error, $\beta$, for Problem~\ref{prob:reworked_prob} is upper bounded by:
    \begin{equation}\label{eq:thm2_eq_statement}
        \beta^2 \leq -\frac{ 4 \tilde \sigma^2}{a q^2} \log \left (1- \left (1-\alpha \right)^2 \left (1-\mathrm{e}^{-\frac{b Q^2}{ 4 \tilde \sigma^2}} \right ) \right),
    \end{equation}
   where $a=\frac{1}{2}$, $b=\frac{2}{\pi}$, and $Q=2\left (\frac{q-1}{2}\right )^2(v+l)+1$ for some constant $l\geq 1$. In particular, the upper bound increases sublinearly as a function of $\tilde \sigma^2=\sigma^2 \Vert T_c^\top d \Vert_2^2$. 
\end{thm}
\begin{proof}
See appendix.
\end{proof}

Theorem~\ref{thm:bounds} does not claim that finding the shortest filtering vector $d$ is simpler for our problem. Rather, it says that the strength of the hypothesis test for detection does not deteriorate as quickly with growing $\Vert d\Vert_2$ as the hypothesis test that solves the LWE problem.

We will not state the corresponding result for Problem~\ref{prob:secondary}. However, we argue that the bound will increase with a similar rate since the density function shown in Proposition~\ref{prop:prop_rejection_rule} is the average of mean-shifted discrete normals, $\mathcal{N}_{\mathbb{Z}_q}(kr,d^\top \Sigma_p d)$, $\vert k\vert  \leq \epsilon_\infty \Vert d_{\epsilon_0} \Vert_1$. In fact Fig.~\ref{fig:T2errs} verifies this claim for a few different parameters. Note how the Type~II Error for revealing the secret key starts close to $1-\alpha$.

\begin{rem}
    At first glance, Theorem~\ref{thm:bounds} seems to violate the security result in~\cite{regev2009} (Lemma~5.4), which states that the LWE problem can be solved if encryptions of $0$ and $M\neq 0$ can be distinguished efficiently. Formally, if we denote the procedure in Section~\ref{sec:results} as $W$, where $W$ outputs 1 if there is an alarm and $0$ otherwise, then the LWE problem can be solved efficiently if  $ \vert \mathrm{Pr}(W\vert m = 0)-\mathrm{Pr}(W \vert m = M)\vert\geq 1/\mathrm{O}(v^c)$, for some $c>0$. However, the binary choice of messages changes $\mathcal{H}_1$ into $\mathcal{N}_{\mathbb{Z}_q}(M,\tilde \sigma^2)$. Since $H_0: \mathcal{N}_{\mathbb{Z}_q}(0,\tilde\sigma^2)$ remains the same, the procedure $W$ remains the same, only the Type~II error changes. Note that, $ \mathrm{Pr}(W\vert m = 0) = \alpha$ is the Type~I error, and $\mathrm{Pr}(W \vert m = M)=1-\beta$, where $\beta$ is the Type~II error. Thus the security is violated if $\vert \beta - (1-\alpha) \vert\geq 1/\mathcal{O}(v^c)$. Showing that $\beta$ approaches $(1-\alpha)$ exponentially, thus not breaking the security, is omitted, however, we shall outline a few steps: First, for any region $\mathcal{S}_{\tilde \sigma^2} \not \ni M $ we have the following bound $\vert \beta - (1-\alpha) \vert\leq \mathcal{O} \left((1-\alpha) - \sum \limits_{k \in \mathcal{S}_{\tilde \sigma^2} } \mathrm{e}^{-\frac{(k-M)^2}{2\tilde\sigma^2}}/q \right)$, second, since $\vert \mathcal{S}_{\tilde \sigma^2} \vert $ grows with $\tilde \sigma^2$, we pick the element $\bar{ s} \subset \mathcal{S}_{\tilde \sigma^2\to \infty}$ closest to $M$, so that  $\vert \beta - (1-\alpha) \vert\leq \mathcal{O} \left( 1- \mathrm{e}^{-\frac{(\bar s-M) ^2}{2\tilde\sigma^2}} \right)$. Since $\tilde \sigma^2 = \Vert d \Vert_2^2 \sigma^2$, since short $d$ are difficult to find in $q$-ary lattices~\cite{ajtai1996}, we can assume they are not sublinear in $v$. Thus, we have that $\vert \beta - (1-\alpha) \vert \not \geq 1/\mathcal{O}(v^c)$.
    \end{rem}

\begin{figure}
\hspace{-0.4cm}\includegraphics[width=0.5\textwidth,trim={0 10 40 10}, clip]{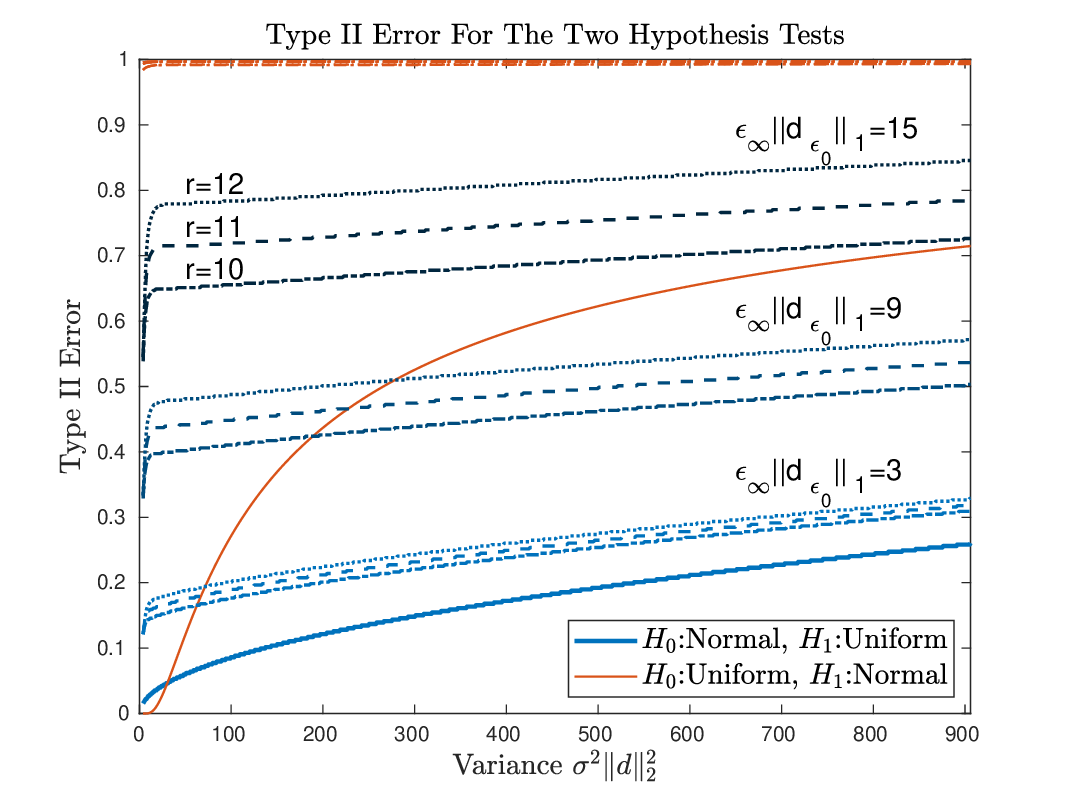}
    \caption{The figure shows how the Type~II Errors change when the plant output is not predicted precisely, for instance, due to quantization errors. The message space is chosen to be $q=453$. While the initial Type~II Errors becomes larger as $r$ and $\epsilon_\infty \Vert d_{\epsilon_0}\Vert_1$ increases, the rate seems to approach the one we proved in Theorem~\ref{thm:bounds}. Finally, note how the Type~II Error for the test that reveals the secret key immediately is close to $1$ for the considered cases.}
    \label{fig:T2errs}
\end{figure}

\section{Controllers, Short Vectors, and Detection} \label{sec:controllers}
This section will discuss obtaining $\Sigma_p$ from the controller implementation and $d$ from a structured short vector problem, both of which are needed for detection. Although multiple ways exist to implement an LWE-compatible controller, such as~\cite{kim2022dynamic}, we will consider a dynamic controller that periodically resets its internal state, $x_k^c = x_{ss}^c$ every $k_p$ time steps, to counteract noise build-up. Let the nominal controller be given by
\begin{equation}\label{eq:controller_real}
    \begin{cases}
    \begin{aligned}
    x^c_{k+1} & =  A_cx^c_k + B_c \bar y_{k}, \\
     u_k & =   C_cx^c_k + D_c \bar y_{k}. 
    \end{aligned}
    \end{cases}
\end{equation}
For instance, such a controller can be obtained by combining an optimal feedback law with a Kalman filter. One may also add states that do not affect $u_k$ by properly choosing extended matrices $A_c$ and $C_c$. To ensure compatibility with the form in~\eqref{eq:controller_estimates_output}, we assume that the output can be estimated from the controller states and, potentially, a direct term scaled by $D_0$:
\begin{equation} \label{eq:estimate_output}
    \hat y_k= C_0x_k^c+D_0\bar y_{k}.
\end{equation}
Implementing a real controller under the LWE-based encryption scheme requires it to work with integer representations of signals and systems. One controller in this paper will be a special case of the one introduced in~\cite{murguia2020}, namely
\begin{equation}\label{eq:integer_controller}
    \begin{cases}
    \begin{aligned}
    c^{\bar k+1}x^c_{k+1} & =  \lceil c A_c \rfloor c^{\bar k} x^c_k + \lceil cB_c \rfloor c^{\bar k} y^e_{k} \bmod q, \\
     c^{\bar k+1} u^e_k & =   \lceil c C_c \rfloor c^{\bar k} x^c_k + \lceil c D_c \rfloor c^{\bar k} y^e_{k} \bmod q, \\
     c^{\bar k+1} \hat{y}^e_{k} & = \lceil c  C_o \rfloor c^{\bar k} x_k^c+\lceil c  D_o \rfloor c^{\bar k} y_{k}^e \bmod q, \\
     x^c_k & = x^c_{ss}, \, \text{when } \bar k= k - \left\lfloor k/k_p\right\rfloor k_p = 0.
    \end{aligned}
    \end{cases}
\end{equation}
Note that the difference from the one in~\cite{murguia2020} is that~\eqref{eq:integer_controller} has a different basis, quantified by scaling factor $c$ instead of a binary representation, and it operates over negative numbers.

Let us briefly explain this controller implementation. Since the controller must be compatible with the integer representation of signals, we approximate the controller matrices by rounding after scaling with some factor $c>1$, $A_c \rightarrow \frac{\lceil cA_c \rfloor}{c}$ and $B_c \rightarrow \frac{\lceil cB_c \rfloor}{c}$. The control state update becomes:
\begin{align*}
    x_{k+1}^c & = A_cx^c_k + B_c \bar y_{k} \approx \frac{1}{c}\left ( \lceil c A_c \rfloor x^c_k + \lceil cB_c \rfloor \bar y_{k} \right ), \\
     & \Rightarrow cx_{k+1}^c \coloneqq \lceil c A_c \rfloor x^c_{k} + \lceil cB_c \rfloor \bar y_{k}.
\end{align*}

The last line comes from the restriction that division is not allowed. The compatible controller, therefore, scales its updated state, $cx_{k+1}^c$, and is used in the next iteration. The measurement, $\bar y_{k+1}$, must also be scaled with $c$ to ensure consistency with the magnitudes
\begin{equation*}
    c^2x_{k+2}^c= \lceil c A_c \rfloor c x^c_{k+1} + \lceil cB_c \rfloor c \bar y_{k+1}.
\end{equation*}
Thus, when simulating integer representation of systems, the dynamical state grows quickly for every iteration, even if the simulated system is stable. Eventually, an overflow will occur after many multiplications, known as the \emph{multiplicative depth}. The last line in~\eqref{eq:integer_controller} is therefore introduced to reset the controller's state to a value (typically an integer) that can be represented without scaling. Furthermore, the state reset also removes any accumulated noise built up due to the noise injection in the LWE-cryptosystem.

Finally, the other signals of controller~\eqref{eq:integer_controller}, $u^e_k$ and $y^e_k$, and the matrices $C_c$, $D_c$, $ C_o$, and $ D_o$, are treated analogously.

\begin{rem}
    The scaling factor $c^{\bar k+1}$ in the encrypted control messages, $c^{\bar k+1}u^e_k$, implies that the plant has to downscale the message with the corresponding factor after decryption. Therefore, time-step coordination between the plant and the controller may be required. 
\end{rem}
\begin{figure*}
\begin{equation} \label{eq:long_dyn_cont}
    \begin{bmatrix}
        c \hat y^e_k \\ c^2 \hat y_{k+1}^e \\ \vdots \\ c^{k_p}\hat y_{k+k_p-1}^e
    \end{bmatrix}  
    = \begin{bmatrix}
        \lceil c D_o \rfloor & \textbf{0} & \cdots & \textbf{0} \\
        \lceil c C_o \rfloor \lceil c B_c \rfloor & \lceil c D_o \rfloor c & \cdots & \textbf{0}  \\
        \vdots & \vdots &  \ddots &  \vdots \\
         \lceil c C_o \rfloor \lceil c A_c \rfloor^{k_p-2} \lceil c B_c \rfloor &  \lceil c C_o \rfloor \lceil c A_c \rfloor^{k_p-3} \lceil c B_c \rfloor c & \cdots &  \lceil c D_o \rfloor c^{k_p-1}
    \end{bmatrix} \begin{bmatrix}
         y^e_{k} \\  y_{k+1}^e \\ \vdots \\  y_{k+k_p-1}^e
    \end{bmatrix} 
    = \mathcal{T} Y^c_{k:k+k_p-1}.
\end{equation}
\hrule
\end{figure*}

To perform the anomaly detection in Theorem~\ref{thm:main_theorem_rejection_rule}, the covariance matrix to the output estimate of~\eqref{eq:integer_controller} is required. Let us, for simplicity, assume that $x_{ss}^c=0$. Then, assuming that the modulus operator is not triggered, we can write the output trajectory between resets as in Equation~\eqref{eq:long_dyn_cont}, where $\mathcal{T}$ denotes the Toeplitz matrix of the controller parameters, and $Y^c_{k-1:k+k_p-2}$ denotes the scaled encrypted measurements. Running the controller for multiple periods of length $k_p$ results in the encrypted trajectory:
\begin{equation*}
    \hat Y^c =\begin{bmatrix}
        \mathcal{T} & \textbf{0} & \cdots & \textbf{0} \\
        \textbf{0} & \mathcal{T} & \cdots & \textbf{0} \\
        \vdots & \vdots & \ddots & \vdots \\
        \textbf{0} & \textbf{0} & \cdots & \mathcal{T}
    \end{bmatrix} Y^c = (I \otimes \mathcal{T})Y^c,
\end{equation*}
for the periodically scaled encrypted output trajectory $Y^c$ and the expected encrypted trajectory $\hat Y^c$ with a length that is a multiple of $k_p$, $N=lk_p$, for $l \in \mathbb{N}$. The scaled encrypted residual trajectory in~\eqref{eq:encrypted_residual} is then:
\begin{equation}\label{eq:transform_encrypt}
   \rho^c = (I \otimes \mathcal{T})Y^c-Y^c = (I \otimes (\mathcal{T}-I))Y^c.
\end{equation}
The scaled residual transforms the scaled encrypted output in a way determined by the controller's and plant's dynamics through $\mathcal{T}$.

Now, assume that $Y^c \sim \mathcal{N}_{\mathbb{Z}_q}(\hat Y^e,\Sigma)$. Then, transformation~\eqref{eq:transform_encrypt} says that:
\begin{equation*}
    \rho^c \sim \mathcal{N}_{\mathbb{Z}_q}(\textbf{0},\Sigma_p), 
\end{equation*}
where $\Sigma_p = (I \otimes (\mathcal{T}-I)) \Sigma (I \otimes (\mathcal{T}-I))^\top=T_c \Sigma T_c^\top$. We must use this covariance matrix when testing for attacks in Theorem~\ref{thm:main_theorem_rejection_rule}. By designing the controller, one may influence $\mathcal{T}$, and by extension, the variance in the hypothesis test. Furthermore, the controller reset every $k_p$ time step induces a periodic error in the estimator, so that $\epsilon_0, \epsilon_\infty \neq 0$.


\subsection{Detection Strength From Short Vectors}
\label{sec:vectors}

Let us now address the problem of finding short filtering vectors, in terms of small $\Vert d \Vert_2$. We have seen in Sections~\ref{sec:hyptest} how larger norms negatively impacts the detection capabilities. To obtain short filtering vectors we have to apply a two step process; 1)~obtaining initial (typically long) vectors, and 2)~shortening the obtained vectors. The initial vectors can be obtained as a "by-product" of the HNF to $\mathcal{P}$, as mentioned in Definition~\ref{def:qaryLattice}. In particular, for sufficiently large $N$ we will have $N-Lpv$ filtering vectors  $d_i\in \Ker_q \mathcal{P}$, $i \in {1, \, \dots, \,  N-Lpv}$, with high probability, which can be computed with the reciprocal of the Riemann Zeta function. We place these vectors as columns in $\mathcal{D}= \begin{bmatrix}
         d_1 & \dots &  d_{N-Lpv}
    \end{bmatrix}$. Then 
\begin{equation} \label{eq:HNF}
 \mathcal{P} \begin{bmatrix}
        U & \mathcal{D}
    \end{bmatrix} = \begin{bmatrix}
        \mathcal{P} U & \mathcal{P} \mathcal{D}
    \end{bmatrix} = \begin{bmatrix}
        H_{\mathcal{P}} & 0_{v \times (N-Lpv)}
    \end{bmatrix},
\end{equation}
     the HNF of $\mathcal{P}$, $H_{\mathcal{P}} \in \mathbb{Z}_z^{v \times v}$, is a full-rank lower triangular matrix, and $\begin{bmatrix}
        U & \mathcal{D}
    \end{bmatrix}$ is an unimodular matrix. The HNF will typically yield long column vectors in $\mathcal{D}$. However, since
    \begin{equation*}
    \begin{aligned}
        0 & = (c_id_i^\top\mathcal{P} \bmod q) +(c_jd_j^\top\mathcal{P} \bmod q ) \bmod q \\
        & =  (c_id_i +c_jd_j)^\top\mathcal{P} \bmod q,
        \end{aligned}
    \end{equation*}
a search for shorter filtering vectors can be done by considering integer-linear combinations of the columns in $\mathcal{D}$. Since only integer-linear combinations are considered, the columns in $\mathcal{D}$ can be treated as a basis for a lattice $\mathcal{L}$, and the search for short vector is equivalent to finding short vectors in $\mathcal{L}$.

\begin{rem}
    Typically, and especially for when $q$ is prime, one can find a $\mathcal{D}$ so that its diagonal elements are only ones. One can the fully express a lattice in $\mathbb Z_q$ by complementing $\mathcal D$ with the matrix $\mathcal{D}_q$ whose columns only contain one non-zero element, $q$, so that $\bar{\mathcal{D}}= \begin{bmatrix}
        \mathcal{D} & \mathcal{D}_q
    \end{bmatrix}$ has full rank and its diagonal is first composed of ones then of elements with $q$.
\end{rem}

Finding short vectors can be done using the LLL algorithm. While the LLL algorithm runs in polynomial time, the shortest vector it outputs is only guaranteed to be within an \emph{exponentially} close distance to the shortest vector in the lattice,
    \begin{equation} \label{eq:first_basis_length}
    \min_{d \in \text{LLL}(\mathcal{D})} \Vert d \Vert_2 \leq \frac{1}{\left (\delta-\frac{1}{4}\right )^{\frac{N-Lpv}{2}}} \min_{d \in \mathbb{L}} \Vert d \Vert_2,
\end{equation}
where the Lov\'ascz number $\delta$ is $\frac{1}{4}<\delta < 1$. Furthermore, the output is sensitive to the order of $\mathcal{D}$'s columns, leading to combinatorially many different potential initial bases.

However, recall that we are testing on $\vert d^\top \rho^e \vert$, which implies that if the $k$:th element in $d$ is non-zero, then $\rho_k^e$ is included in detection scheme. Thanks to this physical meaning attached to the elements of $d$, we can order the intial basis so that the search in the LLL among short vectors is done by using the most recent samples first.

\section{Numerical results} \label{sec:numerical}

We will now construct an LWE encrypted control setup that allows for anomaly detection, and perform anomaly detection over windows of samples to show how Theorem~\ref{thm:main_theorem_rejection_rule} can be used to detect sensor bias attacks over encrypted signals. The length of the detection windows will be determined by how often we switch the secret vector $s$, which in turn depends on the security level of the cryptosystem scheme and the assumed strength of the attacker.


Consider the dynamical system~\eqref{eq:system}, with the parameters
\begin{equation} 
\begin{aligned}
    A & =\begin{bmatrix}
     1.001 &  0.4  \\
      0 &  0.1 
    \end{bmatrix}, \, B = \begin{bmatrix}
            0 \\
     0.4722 
    \end{bmatrix}, \\
     C & = \begin{bmatrix}
        1 &  0
    \end{bmatrix}, \, D = 0, \, \Sigma_w  =0, \, \Sigma_v=0,
\end{aligned} \label{eq:numerical_system}
\end{equation}
which is an unstable system that needs to be stabilized by a quantized controller. In particular, we consider quantizing the dynamic controller of the form in~\eqref{eq:integer_controller}
with
\begin{equation*}
\begin{aligned}
    \lceil cA_c \rfloor & = \begin{bmatrix}
     1 &  1 & 0  \\
      -1 & -1 & 1 \\
      0 & 0 & 2
    \end{bmatrix}, & \lceil cB_c \rfloor & = \begin{bmatrix}
            1 \\
            -2 \\
     -2 \end{bmatrix}, \\
     \lceil cC_c \rfloor & = \begin{bmatrix}
        -2 &-3 &  1
    \end{bmatrix}, &  \lceil c D_c \rfloor & = 0, \\ 
    \lceil cC_o \rfloor & = \begin{bmatrix}
        1 & 0 &  0
    \end{bmatrix}, &  \lceil cD_o \rfloor & = 1,
\end{aligned}
\end{equation*}
The quantization scheme is implemented using $c_s=5$, and $c=2$, which are compatible with system~\eqref{eq:numerical_system} and the LWE-based encryption with parameters $\lambda=(64,300,10,2^{16}+1)$, where a reset is performed every $k_p=4$ time steps to $x^c_{ss}=\begin{bmatrix}
    0 & 0 & 0
\end{bmatrix}^\top$. This reset horizon is sufficiently long to ensure that the resulting closed-loop system is stable. The parameters allow us to compute $T_c$ from~\eqref{eq:long_dyn_cont} which will be used in the detection mechanism.

Using the Lattice Estimator~\cite{Albrecht2015}, the chosen cryptosystem parameters correspond to a $39$-bit security level once more than $2v = 128$ samples are collected. A (high-end) 9 GHz CPU would take approximately $\frac{2^{39}}{9 \cdot 10^9} \approx 60$ seconds to crack a cryptosystem with that level of security. It is, therefore, required that the secret key is reset not later than approximately a minute after 128 samples have been disclosed. However, Theorem~\ref{thm:bounds} states that the detection becomes better when shorter vectors are used in the detection mechanism. Fig.~\ref{fig:length_structure_tradeoff} shows how the LLL algorithm outputs shorter vectors the more samples are included in the lattice problem. Although the decay is monotonical, it slows down after a while, indicating that more samples will only have a marginal impact. However, the LLL algorithm's run time will increase, highlighting the trade-off between detection strength (the number of vectors) and computational time.


\begin{figure}
\centering
    \includegraphics[width=1\linewidth,trim={30pt 0pt 40pt 0pt},clip]{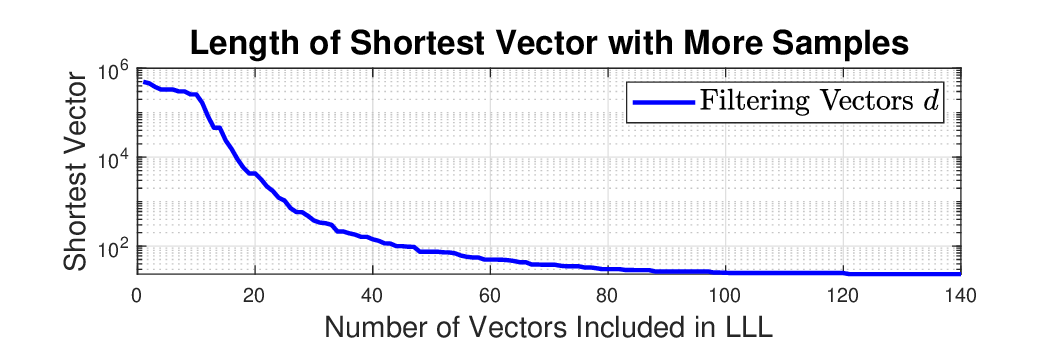}
    \caption{The LLL algorithm outputs shorter vectors as more samples are included, but the decay slows down, indicating marginal impact from additional samples. Here we have used vectors from the LWE problem with  $v = 10$ and $q = 227$.}
    \label{fig:length_structure_tradeoff}
\end{figure}



We will assume that the closed-loop system runs another $v+k_p=64+4$ time steps during that minute (resulting in a sampling time of 0.88 Hz) before performing a secret key reset. In summary, the secret key is reset every $3v+k_p=196$ time steps to ensure that the attacker cannot uncover the secret key during its use, while also providing the detector with as many samples as possible to enable detection with a high detection strength.



Now consider the system with the initial state $x_0= \begin{bmatrix}
    0.1 & 0
\end{bmatrix}^\top$, which is stabilized around $x_{ss}= \begin{bmatrix}
    0 & 0
\end{bmatrix}^\top$. An attacker initializes a sensor bias attack at $k_{a}=800$, where a constant $y_a=r$ is injected persistently. The injection is possible due to the homomorphic property of the LWE cryptosystem, which results in the following signal being sent to the controller
\begin{equation*}
    y_k^e = \begin{cases} \begin{aligned}  &  \enc_{s_k}(y_k),  &\text{for } k < k_{a}, \\ &  \enc_{s_k}(y_k) + (0, r) \bmod q,  &\text{for } k_{a} \leq k,
    \end{aligned} \end{cases} 
\end{equation*}
where we have $(P_k, P_ks_k + r\lceil c y_k \rfloor + e_k ) + (0, r) \bmod q$ in the second line. Equivalently, the controller receives the encrypted message $\enc_{s_k}(y_k+1)$ from $k=k_{a}$, and onward. 

Fig.~\ref{fig:esrsv} shows the encrypted residual of the system, including changes of the secret vector $s$ every $196$ time steps. The message part of the purely encrypted signals seems uniformly random, and - due to the pseudo-random property - it will pass any tests for uniformity. Furthermore, it is not possible to use ordinary statistical methods to detect that an attack is present from time step $k=800$ and onwards. Applying the test from Theorem~\ref{thm:main_theorem_rejection_rule}, we obtain an aggregate of the residuals that we can perform detection on. Choosing false alarm rates of $1\%$, $5\%$, and $32\%$, respectively, we get the thresholds shown in Fig.~\ref{fig:detection}. One can see, for instance, that all detectors react to the initial condition being different from the assumed one of $x_{ss} = \begin{bmatrix}
    0 & 0
\end{bmatrix}^\top$. Thereafter, two out of three detectors react to when the attack is initialized, only to gradually start missing the subsequent detections. Looking back at Fig.~\ref{fig:esrsv}, we observe that $\sigma^2\Vert T_c^\top d \Vert_2^2 \approx 5 \cdot 10^7$, which for the $5\%$ false alarm threshold puts the Type II error at $50\%$, meaning that approximately half of the attacked windows are missed by the detector, which the numerical results confirm.

\begin{figure}
    \centering
    \includegraphics[width=1\linewidth, trim=30pt 70pt 40pt 70pt, clip]{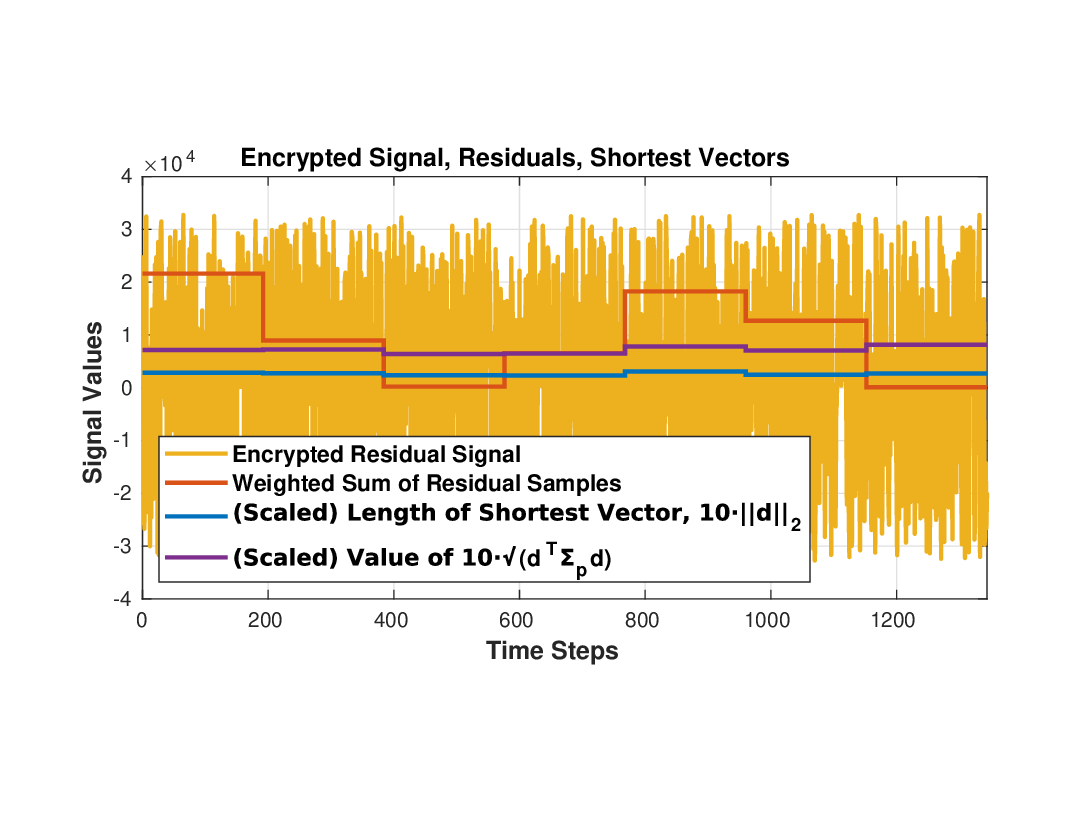}
    \caption{The figure displays the encrypted residual, where the observer state resets every $k_p=4$ time step, with the secret key $s$ refreshing every $3v+k_p=196$ time steps. Only samples with the same secret key can be used for detection, and the figure shows the shortest vector (scaled) that the PotLLL algorithm outputs for each such window.}
    \label{fig:esrsv}
\end{figure}


We would like to end the numerical results by mentioning that the $50\%$ Type II error for the detector with $5\%$ false alarm probability was achieved when the objective was to minimize $\Vert d \Vert_2^2$. A stronger detector would be obtained if $\Vert T_c^\top d \Vert_2^2$ was minimized instead. Initial naïve searches show that a factor of $5$ could be removed from the resulting smallest $\Vert T_c^\top d \Vert_2^2$, indicating that a potential to obtain a Type II error of $10\%$ could be achieved with a different lattice reduction approach.


\begin{figure}
    \includegraphics[width=1\linewidth, trim=30pt 75pt 40pt 78pt, clip]{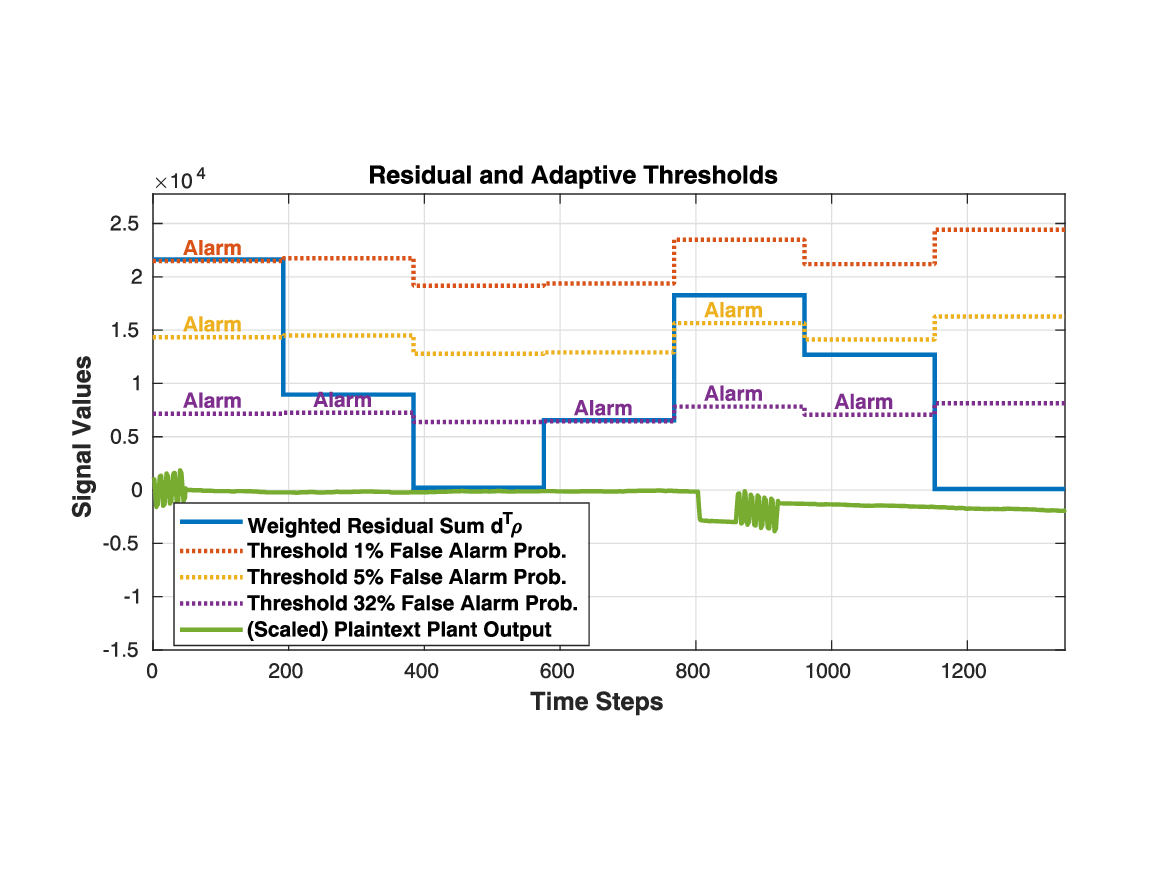}
    \caption{The graph displays the computed residuals for seven different windows using the proposed scheme, along with threshold values for 1\%, 5\%, and 32\% false alarm probabilities. The system state is initialized at a non-zero value, and the controller stabilizes it around zero within the first window, as visible in the plaintext output. Alarms are raised for all three thresholds until the observer's state matches the plant's. At 80 seconds, an attack affects the plaintext output, resulting in generally higher residuals, which demonstrates the feasibility of the proposed scheme as an anomaly detection mechanism. In the non-zero residual windows, $\sigma^2\Vert T_c^\top d \Vert_2^2$ ranges between $5 \times 10^7$ and $9 \times 10^7$ (not shown), causing the 5\% false alarm threshold to miss approximately half of the detections.}
    \label{fig:detection}
\end{figure}

\section{Conclusions}\label{sec:conclusions}
This paper has considered anomaly detection over LWE-based encrypted signals without using secret keys, evaluation keys, or overhead in terms of multi-round communication. In particular, the detector relies on using the controller's intrinsic state estimator to output a prediction of the output based on the input signal it computes. The anomaly detector compares the predicted output with the actual output to create a particular aggregate of the residual that removes the dependence on the secret and public keys. Furthermore, we use the same methods as the original LWE paper~\cite{regev2009}, to show that the original security is not compromised in our setup.


We show how the detection mechanism depends on a vector reduction problem, similar to the security mechanism. Detection over encrypted signals does not require very short vectors; moderately short vectors suffice. However, our proposed detector is sensitive to errors made by the controller's estimator. We provide an initial analysis of how these errors affect detection performance, but further development is needed. The combination of unstable systems and control over quantized signals can lead to limit cycle behavior, including chaos. Co-designing the controller, quantization, and cryptosystem to handle these situations and avoid chaos is an exciting topic for future research, enabling detection over encryption for a broader range of systems.

The search for short vectors, performed here using the (Potential) LLL algorithm, can be extended using stronger lattice reduction algorithms, such as the (Block-)Korkine-Zoltarev algorithms. Although these algorithms are highly sequential, the windowed detection proposed in this paper allows for a degree of parallelism. An interesting future direction is to perform weighted lattice reduction, taking into account the noise and dynamics collected in $\Sigma_p$. We have explored the trade-off between statistical power and false alarm rates. Future work will consider additional trade-offs, such as the frequency of state resets and different controller implementations in the encrypted scheme. Another direction is to prioritize less noisy samples of the residual during vector reduction and to use multiple filtering vectors simultaneously. Transformations of the encrypted residual to improve detectability or reduce detection time will also be considered.

\section*{References}
\bibliographystyle{ieeetr}
\bibliography{bibliography.bib}

\begin{thebibliography}{10}

\bibitem{muller2018}
M.~I. Müller, J.~Milošević, H.~Sandberg, and C.~R. Rojas, ``A risk-theoretical approach to $\mathcal{H}_{2}$-optimal control under covert attacks,'' in {\em 2018 IEEE Conference on Decision and Control (CDC)}, pp.~4553--4558, 2018.

\bibitem{park2016}
G.~{Park}, H.~{Shim}, C.~{Lee}, Y.~{Eun}, and K.~H. {Johansson}, ``When adversary encounters uncertain cyber-physical systems: Robust zero-dynamics attack with disclosure resources,'' in {\em 2016 IEEE 55th Conference on Decision and Control (CDC)}, pp.~5085--5090, 2016.

\bibitem{Coulson2019}
J.~{Coulson}, J.~{Lygeros}, and F.~{Dörfler}, ``{Data-Enabled Predictive Control: In the Shallows of the DeePC},'' in {\em 2019 18th European Control Conference (ECC)}, pp.~307--312, 2019.

\bibitem{proctor2016}
J.~L. Proctor, S.~L. Brunton, and J.~N. Kutz, ``Dynamic mode decomposition with control,'' {\em SIAM Journal on Applied Dynamical Systems}, vol.~15, no.~1, pp.~142--161, 2016.

\bibitem{wang2019}
J.~{Wang} and G.~{Yang}, ``Data-driven methods for stealthy attacks on tcp/ip-based networked control systems equipped with attack detectors,'' {\em IEEE Trans. on Cybernetics}, vol.~49, no.~8, pp.~3020--3031, 2019.

\bibitem{alisic2021ecc}
R.~Alisic and H.~Sandberg, ``Data-injection attacks using historical inputs and outputs,'' in {\em 2021 European Control Conference (ECC)}, pp.~1399--1405, 2021.

\bibitem{taheri2021}
M.~Taheri, K.~Khorasani, I.~Shames, and N.~Meskin, ``Data-driven covert-attack strategies and countermeasures for cyber-physical systems,'' in {\em 2021 60th IEEE Conference on Decision and Control (CDC)}, p.~4170–4175, IEEE Press, 2021.

\bibitem{adachi2023}
R.~Adachi and Y.~Wakasa, ``Design of inputs without data informativity for secure model predictive control,'' in {\em 2023 IEEE/SICE International Symposium on System Integration (SII)}, pp.~1--4, 2023.

\bibitem{TEIXEIRA2015}
A.~Teixeira, I.~Shames, H.~Sandberg, and K.~H. Johansson, ``A secure control framework for resource-limited adversaries,'' {\em Automatica}, vol.~51, pp.~135 -- 148, 2015.

\bibitem{weerakkody2015}
S.~Weerakkody and B.~Sinopoli, ``Detecting integrity attacks on control systems using a moving target approach,'' in {\em 2015 54th IEEE Conference on Decision and Control (CDC)}, pp.~5820--5826, 2015.

\bibitem{gheitasi2022}
K.~Gheitasi and W.~Lucia, ``Undetectable finite-time covert attack on constrained cyber-physical systems,'' {\em IEEE Transactions on Control of Network Systems}, vol.~9, no.~2, pp.~1040--1048, 2022.

\bibitem{sandberg2015}
H.~Sandberg, G.~Dán, and R.~Thobaben, ``Differentially private state estimation in distribution networks with smart meters,'' in {\em 2015 54th IEEE Conference on Decision and Control (CDC)}, pp.~4492--4498, 2015.

\bibitem{cortes2016}
J.~Cortés, G.~E. Dullerud, S.~Han, J.~Le~Ny, S.~Mitra, and G.~J. Pappas, ``Differential privacy in control and network systems,'' in {\em 2016 IEEE 55th Conference on Decision and Control (CDC)}, pp.~4252--4272, 2016.

\bibitem{Hassan2020}
M.~U. Hassan, M.~H. Rehmani, and J.~Chen, ``Differential privacy techniques for cyber physical systems: A survey,'' {\em IEEE Communications Surveys \& Tutorials}, vol.~22, no.~1, pp.~746--789, 2020.

\bibitem{giraldo2018}
J.~Giraldo, D.~Urbina, A.~Cardenas, J.~Valente, M.~Faisal, J.~Ruths, N.~O. Tippenhauer, H.~Sandberg, and R.~Candell, ``A survey of physics-based attack detection in cyber-physical systems,'' {\em ACM Comput. Surv.}, vol.~51, no.~4, 2018.

\bibitem{kim2022dynamic}
J.~Kim, H.~Shim, and K.~Han, ``Dynamic controller that operates over homomorphically encrypted data for infinite time horizon,'' {\em IEEE Transactions on Automatic Control}, 2022.

\bibitem{AES2001}
M.~Dworkin, E.~Barker, N.~J, J.~Foti, L.~Bassham, E.~Roback, and J.~Dray, ``{Advanced Encryption Standard (AES)},'' 2001-11-26 2001.

\bibitem{kogiso2018}
K.~Kogiso, ``Attack detection and prevention for encrypted control systems by application of switching-key management,'' in {\em 2018 IEEE Conference on Decision and Control (CDC)}, pp.~5032--5037, 2018.

\bibitem{Boenisch2023}
F.~Boenisch, A.~Dziedzic, R.~Schuster, A.~S. Shamsabadi, I.~Shumailov, and N.~Papernot, ``When the curious abandon honesty: Federated learning is not private,'' in {\em 2023 IEEE 8th European Symposium on Security and Privacy (EuroS\&P)}, pp.~175--199, 2023.

\bibitem{paillier1999}
P.~Paillier, ``Public-key cryptosystems based on composite degree residuosity classes,'' in {\em Advances in Cryptology --- EUROCRYPT '99} (J.~Stern, ed.), pp.~223--238, Springer Berlin Heidelberg, 1999.

\bibitem{regev2009}
O.~Regev, ``On lattices, learning with errors, random linear codes, and cryptography,'' {\em J. ACM}, vol.~56, sep 2009.

\bibitem{ajtai1996}
M.~Ajtai, ``Generating hard instances of lattice problems (extended abstract),'' in {\em Proceedings of the Twenty-Eighth Annual ACM Symposium on Theory of Computing}, STOC '96, (New York, NY, USA), p.~99–108, Association for Computing Machinery, 1996.

\bibitem{alexandru2022}
A.~B. Alexandru, L.~Burbano, M.~F. Çeliktuğ, J.~Gomez, A.~A. Cardenas, M.~Kantarcioglu, and J.~Katz, ``Private anomaly detection in linear controllers: Garbled circuits vs. homomorphic encryption,'' in {\em 2022 IEEE 61st Conference on Decision and Control (CDC)}, pp.~7746--7753, 2022.

\bibitem{buns2018}
B.~Bünz, J.~Bootle, D.~Boneh, A.~Poelstra, P.~Wuille, and G.~Maxwell, ``Bulletproofs: Short proofs for confidential transactions and more,'' in {\em 2018 IEEE Symposium on Security and Privacy (SP)}, pp.~315--334, 2018.

\bibitem{alisic2023modelfreelwe}
R.~Alisic, J.~Kim, and H.~Sandberg, ``{Model-Free Undetectable Attacks on Linear Systems Using LWE-Based Encryption},'' {\em IEEE Control Systems Letters}, vol.~7, pp.~1249--1254, 2023.

\bibitem{lenstra1982}
A.~K. Lenstra, H.~W. Lenstra, and L.~Lov{\'a}sz, ``Factoring polynomials with rational coefficients,'' {\em Mathematische annalen}, vol.~261, no.~ARTICLE, pp.~515--534, 1982.

\bibitem{kannan1979}
R.~Kannan and A.~Bachem, ``{Polynomial Algorithms for Computing the Smith and Hermite Normal Forms of an Integer Matrix},'' {\em SIAM Journal on Computing}, vol.~8, no.~4, pp.~499--507, 1979.

\bibitem{boneh2011}
D.~Boneh and D.~M. Freeman, ``Linearly homomorphic signatures over binary fields and new tools for lattice-based signatures,'' in {\em Public Key Cryptography -- PKC 2011} (D.~Catalano, N.~Fazio, R.~Gennaro, and A.~Nicolosi, eds.), (Berlin, Heidelberg), pp.~1--16, Springer Berlin Heidelberg, 2011.

\bibitem{delchamps1990}
D.~Delchamps, ``Stabilizing a linear system with quantized state feedback,'' {\em IEEE Transactions on Automatic Control}, vol.~35, no.~8, pp.~916--924, 1990.

\bibitem{Pasqualetti2013}
F.~{Pasqualetti}, F.~{Dörfler}, and F.~{Bullo}, ``Attack detection and identification in cyber-physical systems,'' {\em IEEE Transactions on Automatic Control}, vol.~58, no.~11, pp.~2715--2729, 2013.

\bibitem{Lehmann2005}
E.~L. Lehmann and J.~Romano, {\em Unbiasedness: Theory and First Applications}, pp.~110--149.
\newblock New York, NY: Springer New York, 2005.

\bibitem{rotenberg60}
A.~Rotenberg, ``A new pseudo-random number generator,'' {\em J. ACM}, vol.~7, p.~75–77, jan 1960.

\bibitem{murguia2020}
C.~Murguia, F.~Farokhi, and I.~Shames, ``Secure and private implementation of dynamic controllers using semihomomorphic encryption,'' {\em IEEE Transactions on Automatic Control}, vol.~65, no.~9, pp.~3950--3957, 2020.

\bibitem{Albrecht2015}
M.~R. Albrecht, R.~Player, and S.~Scott, ``On the concrete hardness of learning with errors,'' {\em Journal of Mathematical Cryptology}, vol.~9, no.~3, pp.~169--203, 2015.

\bibitem{Chu1955}
J.~T. Chu, ``On bounds for the normal integral,'' {\em Biometrika}, vol.~42, no.~1/2, pp.~263--265, 1955.

\end{thebibliography}

\section*{Appendix: Proof of Theorem~\ref{thm:bounds}}
Let us start by dealing with the modulo operation. In particular, sums of samples from the discrete normal on $\mathbb Z_q$ may end up on $\mathbb Z_Q$ so we need to account for it. Choosing $\sigma^2$ according to Assumption~\ref{asm:noise_var} implies that, within the statistical distance, we can use a discrete normal with a larger support,
    \begin{equation}\label{eq:temp_approx}
    \left \vert \mathrm{Pr}_{\mathbb{Z}_q}(x)- \mathrm{Pr}_{\mathbb{Z}_Q}(x)\right \vert \leq \mathcal{O}(\epsilon), \quad \text{for } x \in \mathbb{Z}_q.
    \end{equation}
Due to HNF, the vector $d$ will have $v+1$ nonzero elements, all in $\mathbb{Z}_q$. However, combining $l > 1$ vectors to reduce $\Vert d \Vert_2$ will result in $v+l$ non-zero elements. The support of the sum will then be in $\pm \left (\frac{q-1}{2} \right)^2(v+l)$. Consider the following function,\begin{equation}\label{eq:discrete_gaussian_cdf}
        f_{\tilde \sigma^2}(s) \coloneqq \frac{\sum \limits_{x = -\left \lfloor\frac{s}{2}\right \rfloor}^{\left \lfloor\frac{s-1}{2}\right \rfloor} \mathrm{e}^{-\frac{x^2}{2 \tilde \sigma^2}}}{S_Q}, \quad \text{where} \, S_Q=\sum \limits_{x \in \mathbb Z_Q} \mathrm{e}^{-\frac{x^2}{2 \tilde \sigma^2}}.
    \end{equation}
    The Type II Error is given by $\beta = \frac{s}{q}$ where we obtain $s \in \mathbb{Z}_q^+$, by finding an $s$ so that $\sum \limits_{k} f_{\tilde \sigma^2}(s+kq)\leq 1-\alpha$. Our idea is simple: since~\eqref{eq:discrete_gaussian_cdf} is the cumulative distribution function of a \emph{discrete} normal, we approximate it by the \emph{continuous} normal.
    \begin{equation*}
        \frac{\int \limits_{-\frac{s}{2}}^{\frac{s}{2}} \mathrm{e}^{-\frac{x^2}{2\tilde \sigma^2} \mathrm{d}x}}{I_Q} \approx f_{\tilde \sigma^2}(s), \quad \text{where} \, I_Q = \int \limits_{-\frac{Q}{2}}^{\frac{Q}{2}} \mathrm{e}^{-\frac{x^2}{2\tilde \sigma^2} }\mathrm{d}x. 
    \end{equation*}
Note that we get exact equality at $s=Q$. For odd $Q$, the discretization of the integrals leads to the following inequality
    \begin{align*}
        & \sum \limits_{x = -\frac{Q-1}{2}}^{\frac{Q-1}{2}} \mathrm{e}^{-\frac{x^2}{2 \tilde \sigma^2}}- \int \limits_{-\frac{Q}{2}}^{\frac{Q}{2}} \mathrm{e}^{-\frac{x^2}{2\tilde \sigma^2}} \mathrm{d}x \\
        & = 2\left (\sum \limits_{x = 0}^{\frac{q-1}{2}} \mathrm{e}^{-\frac{x^2}{2 \tilde \sigma^2}}- \int \limits_{0}^{\frac{Q}{2}} \mathrm{e}^{-\frac{x^2}{2\tilde \sigma^2}} \mathrm{d}x \right) - 1 \\ 
        & \geq  2\left (\sum \limits_{x = 1}^{\frac{Q-1}{2}} \mathrm{e}^{-\frac{x^2}{2 \tilde \sigma^2}}- \int \limits_{1}^{\frac{Q}{2}} \mathrm{e}^{-\frac{x^2}{2\tilde \sigma^2}} \mathrm{d}x + \right) - \mathrm{e}^{-\frac{1^2 }{2\tilde \sigma^2}} \geq \dots \\
        & \geq 2\left (\sum \limits_{x = \frac{Q-1}{2}}^{\frac{Q-1}{2}} \mathrm{e}^{-\frac{x^2}{2 \tilde \sigma^2}} - \int \limits_{\frac{Q-1}{2}}^{\frac{Q}{2}} \mathrm{e}^{-\frac{x^2}{2\tilde \sigma^2}} \mathrm{d}x\right) - \mathrm{e}^{- \frac{\left(Q-1\right)^2}{8 \sigma^2}} \geq 0,
    \end{align*}
    meaning that $S_Q>I_Q$. The inequalities follow from the convexity of $\mathrm{e}^{-x^2}$ for $x > 0$, and
    \begin{equation} \label{eq:helper_ineq_sum_int}
        \mathrm{e}^{-\frac{x^2}{2\tilde \sigma^2}} + \mathrm{e}^{-\frac{(x+1)^2}{2\tilde \sigma^2}} \geq 2 \int \limits_x^{x+1} \mathrm{e}^{-\frac{\tau^2}{2\tilde \sigma^2}}\mathrm{d}\tau, \quad \text{for } x>0.
    \end{equation}
    For even $Q$, extra terms $\mathrm{e}^{\frac{q^2}{8\tilde \sigma^2}}-2\int \limits_{\frac{q-1}{2}}^{\frac{q}{2}} \mathrm{e}^{-\frac{x^2}{2\tilde \sigma^2}}\mathrm{d}x$ appear, however, they can be removed at the last step with~\eqref{eq:helper_ineq_sum_int}. Furthermore,
    \begin{equation*}
        S_Q\int \limits_{x-1}^{x} \mathrm{e}^{-\frac{\tau^2}{2\tilde \sigma^2}}\mathrm{d}\tau \geq I_Q \mathrm{e}^{-\frac{x^2}{2 \tilde \sigma^2}}, \quad \text{for } x>0,
    \end{equation*}
    since $\int \limits_{x-1}^{x} \mathrm{e}^{-\frac{\tau^2}{2\tilde \sigma^2}}\mathrm{d}\tau \geq \mathrm{e}^{-\frac{x^2}{2 \tilde \sigma^2}}$. Therefore, we have that
    \begin{equation*}
        \frac{\mathrm{e}^{-\frac{x^2}{2 \tilde \sigma^2}}}{S_Q}-\frac{\int \limits_{x-1}^{x} \mathrm{e}^{-\frac{\tau^2}{2\tilde \sigma^2}} \mathrm{d} \tau}{I_Q} \leq 0, \quad \text{for } x>0.
    \end{equation*}
Since the difference is negative, we can use it to bound the difference between the discrete and continuous Gaussian.
\begin{equation} \label{eq:inequality_drivation}
    \begin{aligned}
        0 & = \frac{S_Q}{S_Q}-\frac{I_Q}{I_Q} \leq \frac{2\sum \limits_{x = 0}^{\frac{Q-1}{2}} \mathrm{e}^{-\frac{x^2}{2 \tilde \sigma^2}}}{S_Q} - \frac{2\int \limits_{0}^{\frac{Q-1}{2}} \mathrm{e}^{-\frac{x^2}{2\tilde \sigma^2}} \mathrm{d}x}{I_Q} - \frac{1}{S_Q} \\
        & \leq f_{\tilde \sigma^2}(s) - \frac{\int \limits_{-\frac{s}{2}}^{\frac{s}{2}} \mathrm{e}^{-\frac{x^2}{2\tilde \sigma^2}} \mathrm{d}x}{I_Q}, 
    \end{aligned}
    \end{equation}
    where we have again assumed an odd $Q$. The first inequality comes from removing $-2\int \limits_{\frac{Q-1}{2}}^{\frac{Q}{2}} \mathrm{e}^{-\frac{x^2}{2\tilde \sigma^2}} \mathrm{d}x/I_Q$. If $s$ or $Q$ are \emph{even} integers then extra terms with the following form appear,
    \begin{equation*}
        \frac{\mathrm{e}^{-\frac{y^2}{2 \tilde \sigma^2}}}{S_Q} - \frac{2\int \limits_{y-\frac{1}{2}}^{y} \mathrm{e}^{-\frac{\tau^2}{2\tilde \sigma^2}} \mathrm{d}\tau}{I_Q} \leq 0,
    \end{equation*}
    where $y=Q$ or $y=s$. Since they are negative by the mean value theorem, they can just be removed from~\eqref{eq:inequality_drivation}. 
    
    Next, we use the bounds from~\cite{Chu1955}, to the continuous approximation,
    \begin{align*}
        \frac{(1-\mathrm{e}^{-a\frac{s^2}{4\tilde \sigma^2}})^{\frac{1}{2}}}{(1-\mathrm{e}^{-b\frac{Q^2}{4\tilde \sigma^2}})^{\frac{1}{2}}} \leq \frac{\int \limits_{-\frac{s}{2}}^{\frac{s}{2}} \mathrm{e}^{-\frac{x^2}{2\tilde \sigma^2} \mathrm{d}x}}{I} \leq f_{\tilde \sigma^2}(s) \leq \sum \limits_{k} f_{\tilde \sigma^2}(s+kq). 
    \end{align*}
    Inserting $\sum \limits_{k} f_{\tilde \sigma^2}(s+kq) \leq 1-\alpha$ and solving for $s$ we get,
    \begin{equation*}
         s^2 \leq {\frac{4\tilde \sigma^2}{-a}\log \left( 1- (1-\alpha)^2 \left (1-\mathrm{e}^{-b\frac{Q^2}{4\tilde \sigma^2}}\right )\right)}.
    \end{equation*}
    Observe the sublinearity claim of the Type II Error. In particular, we have $\beta \leq C_0 \sqrt{\tilde \sigma^2}$ initially ($\tilde \sigma^2 \to 0$), which then converges to a constant, $\sqrt{\frac{b}{a}(1-\alpha)^2}$, when $\tilde \sigma^2 \to \infty$.

\end{document}